\documentclass[journal,twocolumn,twoside]{IEEEtran}

\IEEEoverridecommandlockouts
\usepackage{amsfonts}
\usepackage{color}
\usepackage{graphicx}
\usepackage[dvips]{epsfig}
\usepackage{graphics} 
\usepackage{times} 
\usepackage[cmex10]{amsmath} 
\usepackage{amssymb}  
\usepackage{cite}
\usepackage{multirow}
\usepackage[tight,footnotesize]{subfigure}
\usepackage{algorithm}

\def\twon #1{\left\|#1\right\|_2}
\def\onen #1{\left\|#1\right\|_1}

\def\frobn #1{\left\|#1\right\|_{\text{F}}}

\def\abs #1{\left|#1\right|}

\def\st{\text{subject to }}

\def\bC{\mathbb{C}}

\def\bR{\mathbb{R}}

\def\m #1{\boldsymbol{#1}}

\def\cD{\mathcal{D}}
\def\cE{\mathcal{E}}

\def \qed {\hfill \vrule height6pt width 6pt depth 0pt}
\def\bee{\begin{equation}}
\def\ene{\end{equation}}

\def\beq{\begin{eqnarray}}
\def\enq{\end{eqnarray}}
\def\lentwo{\setlength\arraycolsep{2pt}}

\newtheorem{lem}{Lemma}
\newtheorem{rem}{Remark}

\newtheorem{thm}{Theorem}

\newtheorem{exa}{Example}

\def\equ #1{\begin{equation}#1\end{equation}}
\def\equa #1{\begin{eqnarray}#1\end{eqnarray}}
\def\sbra #1{\left(#1\right)}
\def\mbra #1{\left[#1\right]}
\def\lbra #1{\left\{#1\right\}}
\def\diag #1{\text{diag}#1}
\def\tr #1{\text{tr}#1}

\def\rank #1{\text{rank}#1}
\def\st {\text{ subject to }}

\title{A Discretization-Free Sparse and Parametric Approach for Linear Array Signal Processing}
\author{Zai Yang, Lihua Xie, {\em Fellow, IEEE}, and Cishen Zhang
\thanks{Manuscript submitted July 2013; accepted July 2014. The copyright will be transferred to the IEEE.

Z. Yang and L. Xie are with the School of Electrical and Electronic Engineering, Nanyang Technological University, 639798, Singapore (e-mail: \{yangzai, elhxie\}@ntu.edu.sg).

C. Zhang is with the Faculty of Engineering and Industrial Sciences, Swinburne University of Technology, Hawthorn VIC 3122, Australia (e-mail: cishenzhang@swin.edu.au).}}

\begin{document}
\maketitle

\begin{abstract}
Direction of arrival (DOA) estimation in array processing using uniform/sparse linear arrays is concerned in this paper. While sparse methods via {\em approximate} parameter discretization have been popular in the past decade, the discretization may cause problems, e.g., modeling error and increased computations due to dense sampling. In this paper, an {\em exact} discretization-free method, named as sparse and parametric approach (SPA), is proposed for uniform and sparse linear arrays. SPA carries out parameter estimation in the continuous range based on well-established covariance fitting criteria and convex optimization. It guarantees to produce a sparse parameter estimate without discretization required by existing sparse methods. Theoretical analysis shows that the SPA parameter estimator is a large-snapshot realization of the maximum likelihood estimator and is statistically consistent (in the number of snapshots) under uncorrelated sources. Other merits of SPA include improved resolution, applicability to arbitrary number of snapshots, robustness to correlation of the sources and no requirement of user-parameters. Numerical simulations are carried out to verify our analysis and demonstrate advantages of SPA compared to existing methods.
\end{abstract}


\section{Introduction}

Farfield narrowband source localization based on observed snapshots of a sensor array is a major problem in array signal processing \cite{krim1996two,stoica2005spectral}, known also as direction of arrival (DOA) estimation. The difficulty of the problem arises from the fact that the observed snapshots are nonlinear functions of the directions of interest. According to estimation schemes adopted, existing methods for the DOA estimation can be classified into three categories: parametric, nonparametric and semiparametric, which are described as follows.

{\em Parametric} methods explicitly carry out parameter estimation using optimization or other methods. A prominent example is the nonlinear least squares (NLS) method (see, e.g., \cite{stoica2005spectral}), which adopts the least squares criterion and has a strong statistical motivation. However, NLS admits the following two shortcomings: 1) it requires the knowledge of the source number which is typically unavailable, and 2) its global optimum cannot be guaranteed with a practically efficient algorithm due to nonconvexity caused by the nonlinearity nature. Either of them may greatly degrade the parameter estimation performance. MUSIC \cite{schmidt1981signal,stoica1989music} does the parameter estimation by studying the subspace of the data covariance matrix and is a large-snapshot realization of the maximum likelihood (ML) method in the case of uncorrelated sources. But it also requires the source number and does not have reliable performance when the number of snapshots is small or correlation exists between the sources. Other subspace-based methods include ESPRIT and their variants (see \cite{stoica2005spectral} for a complete review). In contrast, a {\em nonparametric} method typically produces a {\em dense} spectrum whose peaks are interpreted as source directions. So an alternative name {\em dense} is also used. Examples include conventional beamformer and MVDR (or Capon's method) (see e.g., \cite{stoica2005spectral}) and recently introduced iterative adaptive approach (IAA) \cite{yardibi2010source} which eliminates to a large extent leakage of the beamformer and is robust to correlation of the sources. But it will be shown in this paper that IAA suffers from resolution limit, especially in the moderate/low SNR regime. While the classification above is consistent with \cite{stoica2005spectral}, MUSIC is sometimes categorized in the literature as a nonparametric/dense method by interpreting the plot of its objective function as its power spectrum whose peaks are the optima of the directions.



{\em Semiparametric} methods have been popular in the past decade which remove the direction variables in the observation model by parameter discretization and transform the nonlinear parameter estimation problem into a sparse signal recovery problem under a linear model, followed by some sparse signal recovery technique and support detection of the sparse solution. In particular, the continuous direction range is approximated by a set of discrete grid points under the assumption that the grid is fine enough such that any of the true sources lies on (practically, close to) some grid point. After that, the knowledge is exploited that each of the expanded vector of source signals [composed of (virtual) source signals from candidate directions on the grid] is a sparse signal since the grid size is greatly larger than the source number, and a sparse solution is sought after. Finally, the directions are retrieved from the support of the sparse solution. Consequently, the semiparametric methods are also named as {\em sparse} methods in contrast to {\em dense}. In principle, a source will be detected once the estimated source signal/source power from some direction is nonzero. Following from the literature of sparse signal representation (SSR) and later developed compressed sensing (CS) \cite{chen1998atomic,candes2006compressive,donoho2006compressed}, $\ell_1$ norm minimization and other sparse signal recovery techniques have been widely used in the semiparametric methods (see, e.g., \cite{malioutov2005sparse,hyder2010direction,tang2011aliasing,wei2012doa,yang2013off,gurbuz2012bearing,hu2013doa, abeida2013iterative,duarte2013spectral,huang2012adaptive,carlin2013directions}). However, their theoretical support based on SSR and CS cannot be applied due to dense discretization. Moreover, the sparse recovery techniques usually require one or more practically unknown parameters, e.g., the noise statistics, the source number or regularization parameters, etc. The semiparametric iterative covariance-based estimation (SPICE) method \cite{stoica2011new,stoica2011spice,stoica2012spice} is a breakthrough of the semiparametric methods, in which covariance fitting criteria are adopted with sound statistical motivation, the source power and noise variance(s) are estimated in a natural manner, no user-parameters are required and a connection to the $\ell_1$ norm minimization is shown. In principle, the semiparametric methods are approximation methods due to the discretization scheme adopted. Furthermore, grid selection remains a major problem since 1) a coarse grid leads to a high modeling error and 2) too dense a grid is computationally prohibitive and might result in computational instability (see, e.g., \cite{austin2010relation,stoica2012sparse}). To alleviate the drawbacks of the discretization, preliminary results have been obtained in \cite{yang2012robustly,yang2013off,hu2012compressed, abeida2013iterative, duarte2013spectral,huang2012adaptive}. Overall, because of the inherent discretization scheme, the performance of semiparametric methods is dependent on the trade off between the discretization grid size and the computational workload.

This paper aims at developing discretization-free techniques for DOA estimation with an affordable computational workload in a common scenario of linear arrays, more specifically, uniform linear arrays (ULAs) and sparse linear arrays (SLAs). We consider stochastic source signals and the same covariance fitting criteria as in SPICE. By exploiting the Hermitian Toeplitz structure in the data covariance matrix, the covariance fitting problem is cast as semidefinite programming (SDP) and solved using off-the-shelf SDP solvers, e.g., SDPT3 \cite{toh1999sdpt3}, in a polynomial time. A postprocessing technique is presented in this paper to retrieve, from the data covariance estimate, the parameters of interest including source locations, source powers and noise variance(s). The proposed method is a {\em sparse} method because it utilizes the same covariance fitting criteria of the semiparametric SPICE method and guarantees to produce a sparse parameter estimate. At the same time, it is a {\em parametric} method since it is proven to be equivalently solving a covariance fitting problem parameterized by the aforementioned parameters. Therefore, the proposed method is named as sparse and parametric approach (SPA). The SPA method differs from the existing sparse/semiparametric methods by no need of discretization. Unlike the existing parametric methods, SPA is based on convex optimization and does not require the source number. Theoretical analysis shows that SPA is a large-snapshot realization of the ML estimation and is statistically consistent (in the number of snapshots) under uncorrelated sources. Other merits of SPA include improved resolution, applicability to arbitrary number of snapshots, robustness to correlation of the sources and no requirement of user-parameters. Numerical simulations are carried out to verify our theoretical results and demonstrate the superior performance of SPA compared to existing methods.

Though most, if not all, of the existing sparse methods can be applied to sensor arrays with arbitrary geometry, problems may occur when applied to the ULA and SLA cases as studied in this paper. In particular, their parameter estimates might suffer from some identifiability problem and not be truly sparse as demonstrated in this paper, which conflicts with the name {\em sparse} and decreases their resolution considerably. As a byproduct, a modified SPICE method, named as SPICE-PP, is presented to rectify the problem by incorporating the postprocessing technique presented in this paper.

In the single-snapshot case, the parameter estimation problem studied in this paper is mathematically equivalent to spectral analysis \cite{stoica2005spectral}, for which discretization-free methods have been developed recently in \cite{candes2013towards,bhaskar2013atomic,tang2012compressed}. It is noted that the results of \cite{candes2013towards,bhaskar2013atomic,tang2012compressed} are based on very different techniques. Furthermore, \cite{candes2013towards} and \cite{tang2012compressed} are mainly focused on the noise-free case (the former on the ULA case and the latter on the SLA case in the language of this paper). \cite{bhaskar2013atomic} studies the ULA case in the presence of i.i.d.~Gaussian noise with known noise statistics. The SPA method proposed in this paper can deal with all of the above scenarios without the knowledge of noise statistics. While this paper is focused on the array processing applications in which the multisnapshot case is of the main interest, it is still not clear whether the techniques in \cite{candes2013towards,bhaskar2013atomic,tang2012compressed} specialized for the single-snapshot case can be extended to other cases.\footnote{After submission of this paper, we have investigated connections between the proposed statistical inference method and the deterministic atomic norm technique in \cite{candes2013towards,bhaskar2013atomic,tang2012compressed} and also extended the latter to the multisnapshot case in \cite{yang2014gridless,yang2014continuous,yang2014exact}.}

Notations used in this paper are as follows. $\bR$, $\bR_+$ and $\bC$ denote the sets of real numbers, nonnegative real numbers and complex numbers, respectively. Boldface letters are reserved for vectors and matrices. For an integer $N$, $[N]$ is defined as the set $\lbra{1,\dots,N}$. $\abs{\cdot}$ denotes the absolute value of a scalar or cardinality of a set. $\onen{\cdot}$, $\twon{\cdot}$ and $\frobn{\cdot}$ denote the
$\ell_1$, $\ell_2$ and Frobenius norms, respectively. $\m{x}^T$, $\m{x}^H$ and $\overline{\m{x}}$ are the matrix transpose, conjugate transpose and complex conjugate of $\m{x}$, respectively. $x_j$ is the $j$th entry of a vector $\m{x}$ and $\m{A}_{jk}$ is the $jk$th entry of a matrix $\m{A}$. Unless otherwise stated, $\m{x}_T$ is a subvector of $\m{x}$ with the index set $T$. For a vector $\m{x}$, $\diag\sbra{\m{x}}$ is a diagonal matrix with $\m{x}$ being its diagonal. $\m{x}\succeq\m{0}$ means $x_j\geq0$ for all $j$. $\tr\sbra{\m{A}}$ denotes the trace of a matrix $\m{A}$. For positive semidefinite matrices $\m{A}$ and $\m{B}$, $\m{A}\geq\m{B}$ means that $\m{A}-\m{B}$ is positive semidefinite. $E\mbra{\cdot}$ denotes expectation and $\widehat{\theta}$ is an estimator of $\theta$. For notational simplicity, a random variable and its numerical value will not be distinguished.

The rest of this paper is organized as follows. Section \ref{sec:problem} describes the problem of array processing with linear arrays. Section \ref{sec:SPA_ULA} presents the proposed SPA method in the ULA case followed by Section \ref{sec:SPA_SLA} on the SLA case. Section \ref{sec:property} introduces theoretical properties of SPA. Section \ref{sec:connectionPrior} discusses its connections to existing methods. Section \ref{sec:simulation} presents our numerical simulations. Section \ref{sec:conclusion} concludes this paper.

\section{Problem Description} \label{sec:problem}
\subsection{Observation Model}
Consider $K$ narrowband farfield sources $s_k$, $k\in\mbra{K}$, impinging on a linear array of omnidirectional sensors from directions $d_k\in\left[-90^{\circ},90^{\circ}\right)$, $k\in\mbra{K}$. We are interested in estimation of the direction vector $\m{d}=\mbra{d_1,\dots,d_K}^T$, known as the direction of arrival (DOA) estimation problem \cite{krim1996two,stoica2005spectral}. Denote $\theta_k=\frac{\sin\sbra{d_k}+1}{2}\in\left[0,1\right)$, $k\in\mbra{K}$. Then we can equivalently estimate $\m{\theta}=\mbra{\theta_1,\dots,\theta_K}^T$ since the relation $\m{d}\leftrightarrow\m{\theta}$ is one-to-one. $\m{\theta}$ is called the frequency parameter and its meaning will be clarified later. In this paper, we concern only with the estimation of $\m{\theta}$ for convenience. According to \cite{krim1996two,stoica2005spectral}, the time delays at different sensors can be represented by simple phase shifts, leading to the observation model:
\equ{\m{y}(t) = \sum_{k=1}^K\m{a}\sbra{\theta_k}s_k(t) +\m{e}(t)= \m{A}(\m{\theta})\m{s}(t)+\m{e}(t),\quad t\in\mbra{N},\label{formu:observation_model1}}
where $t$ indexes the snapshot and $N$ is the snapshot number, $\m{y}(t)\in\bC^M$, $\m{s}(t)\in\bC^K$ and $\m{e}(t)\in\bC^M$ denote the observed snapshot, the vector of source signals and the vector of measurement noise at snapshot $t$, respectively, and $M$ is the number of sensors (some notations will be redefined in the SLA case). $\m{A}(\m{\theta})=\mbra{\m{a}\sbra{\theta_1},\dots,\m{a}\sbra{\theta_K}}$ is the so-called array manifold matrix and $\m{a}\sbra{\theta_k}$ is the steering vector of the $k$th source which is determined by the geometry of the sensor array and will be given later. More compactly, (\ref{formu:observation_model1}) can be written into
\equ{\m{Y} = \m{A}(\m{\theta})\m{S}+\m{E},\label{formu:observation_model}}
where $\m{Y}=\mbra{\m{y}(1),\dots,\m{y}(N)}$, and $\m{S}$ and $\m{E}$ are similarly defined.
The number of sources $K$ is assumed unknown in this paper. So, the objective of the array processing is to estimate the unknown parameter $\m{\theta}$ (or equivalently, the vector of DOAs $\m{d}$) given the sensor measurements $\m{Y}$ and the mapping $\m{\theta}\rightarrow \m{A}(\m{\theta})$.

\subsection{Some Standard Assumptions} \label{sec:assmptions}
We introduce some standard assumptions for the problem formulation and solution. $\m{e}(t)$, $t\in\mbra{N}$, are assumed to be spatially and temporarily white, i.e.,
\equ{E\mbra{\m{e}(t_1)\m{e}^H(t_2)}=\diag\sbra{\m{\sigma}}\delta_{t_1,t_2},}
where $\m{\sigma}=\mbra{\sigma_1,\dots,\sigma_M}^T\in\bR_+^M$ is the noise variance parameter and $\delta_{t_1,t_2}$ is a delta function that equals $1$ if $t_1=t_2$ or $0$ otherwise. The source signals and the noise are assumed to be uncorrelated with each other. Moreover, assume that the source signals are uncorrelated spatially and temporarily, i.e.,
\equ{E\mbra{\m{s}(t_1)\m{s}^H(t_2)}=\diag\sbra{\m{p}}\delta_{t_1,t_2},}
where $\m{p}=\mbra{p_1,\dots,p_M}^T\in\bR_+^M$ denotes the source power parameter. Under the assumptions above, the data snapshots $\lbra{\m{y}(1),\dots,\m{y}(N)}$ are uncorrelated with each other and have the covariance matrix
\equ{\m{R}=E\mbra{\m{y}(t)\m{y}^H(t)}=\m{A}\sbra{\m{\theta}}\diag\sbra{\m{p}} \m{A}^H\sbra{\m{\theta}} + \diag\sbra{\m{\sigma}}. \label{formu:ycovariance}}
It is worthy noting that the spatial uncorrelatedness of the sources may not be satisfied in practice, i.e., two sources $\lbra{s_{k_1}(t)}$ and $\lbra{s_{k_2}(t)}$ can be correlated or even coherent (i.e., completely correlated). However, the method proposed in this paper will be shown to be robust to this assumption.

\subsection{Uniform and Sparse Linear Arrays} \label{sec:lineararray}
Both uniform and sparse linear arrays will be studied in this paper. In the ULA case, the sensors are uniformly spaced with a spacing of $\frac{\lambda}{2}$, where $\lambda$ denotes the wavelength
of the sources. For an $M$-element ULA, the steering vector $\m{a}\sbra{\theta_k}$ of source $k$ has the following form with $i=\sqrt{-1}$:
\equ{\m{a}\sbra{\theta_k}=\mbra{1, e^{i2\pi \theta_k},\dots, e^{i2\sbra{M-1}\pi \theta_k}}^T.}
It is clear that $\theta_k$ is the frequency of the uniformly sampled complex sinusoid $\m{a}\sbra{\theta_k}$, a reason for why $\m{\theta}$ is called the frequency parameter. A well-known result is that up to $M-1$ sources can be detected using an $M$-element ULA (see, e.g., \cite{ottersten1998covariance}).

An SLA can be considered as a ULA with ``missing'' sensors, i.e., an SLA takes only a subset, say $\m{\Omega}$, of the sensors of a ULA. Thus an SLA can be represented by its sensor index set $\m{\Omega}\subset\mbra{M}$. Without loss of generality, we assume that $\m{\Omega}$ is sorted ascendingly with $\Omega_1=1$ and $\Omega_L= M$, where $L=\abs{\m{\Omega}}$ denotes the array size (otherwise, we can redefine $\m{\Omega}$, say $\overline{\m{\Omega}}$, such that $\overline{\Omega}_j=\Omega_j-\Omega_1+1$, $j\in[L]$, and let $M=\overline{\Omega}_L$). Then the steering vector of the SLA $\m{\Omega}$ for source $k$, denoted by $\m{a}_{\m{\Omega}}\sbra{\theta_k}$, is
\equ{\m{a}_{\m{\Omega}}\sbra{\theta_k}=\mbra{e^{i2\pi\sbra{\Omega_1-1} \theta_k},\dots, e^{i2\pi\sbra{\Omega_{L}-1}\theta_k}}^T.}
Denote $\m{\Gamma}_{\m{\Omega}}\in\lbra{0,1}^{L\times M}$ a selection matrix such that the $j$th row of $\m{\Gamma}_{\m{\Omega}}$ contains all $0$s but a single $1$ at the $\Omega_j$th position. It is clear that
\equ{\m{a}_{\m{\Omega}}\sbra{\theta_k}=\m{\Gamma}_{\m{\Omega}}\m{a}\sbra{\theta_k}. \label{formu:steervector_SLA}}
Let
\equ{\m{\cD}=\lbra{m_1-m_2+1: m_1,m_2\in\m{\Omega}, m_1\geq m_2}\subset\mbra{M}. \label{formu:coarray}}
An SLA defined by $\m{\cD}$ is called the coarray of $\m{\Omega}$. $\m{\Omega}$ is called a redundancy array if $\m{\cD}$ defines a ULA, i.e., $\m{\cD}=\mbra{M}$. It is obvious that a ULA is a redundancy array. The maximum number of sources detectable using the array $\m{\Omega}$ is determined by its coarray $\m{\cD}$ \cite{linebarger1993difference,ottersten1998covariance}. In particular, a redundancy array can detect up to $M-1$ sources like a ULA. For a non-redundancy SLA, the maximum number of sources detectable will be less than $M-1$. It is noted that redundancy SLAs are quite common and there generally exists such an array satisfying that $L^2\leq 3\sbra{M-1}$ according to \cite{linebarger1993difference}.

\begin{exa} The SLA $\m{\Omega}=\lbra{1,2,5,7}$ is a redundancy array, where $L=4$ and $M=7$. So maximally $6$ sources can be detected using this $4$-element array.
\end{exa}

\section{SPA with a ULA} \label{sec:SPA_ULA}
\subsection{Covariance Fitting Criteria}
Since a ULA can detect up to $M-1$ sources, the condition $K\leq M-1$ can be considered as {\em a priori} knowledge while the exact value of $K$ is unavailable. The covariance of the data snapshots $\m{Y}$ is given in (\ref{formu:ycovariance}) under the assumptions specified in Subsection \ref{sec:assmptions}. Denote by $\widetilde{\m{R}}=\frac{1}{N}\m{Y}\m{Y}^H$ the sample covariance. When $\widetilde{\m{R}}$ and $\m{R}$ are both invertible we consider the following covariance fitting criterion for
the purpose of parameter estimation (see \cite{stoica2011new,stoica2011spice,stoica2012spice,ottersten1998covariance} and the references therein):
\equ{f_1\sbra{\m{\theta},\m{p},\m{\sigma}}= \frobn{\m{R}^{-\frac{1}{2}}\sbra{\widetilde{\m{R}}-\m{R}}\widetilde{\m{R}}^{-\frac{1}{2}}}^2. \label{formu:criterion1}}
$\m{R}^{-1}$ exists in the presence of noise, i.e., $\sigma_j>0$ for $j\in\mbra{M}$. Then $\widetilde{\m{R}}^{-1}$ exists with probability one if $N\geq M$. According to \cite{ottersten1998covariance,stoica1989reparametrization}, the minimization of the criterion in (\ref{formu:criterion1}) is a large-snapshot realization of the ML estimator. In the case where $\widetilde{\m{R}}$ is singular (it happens when $N<M$) we consider an alternative criterion:
\equ{f_2\sbra{\m{\theta},\m{p},\m{\sigma}} =\frobn{\m{R}^{-\frac{1}{2}}\sbra{\widetilde{\m{R}}-\m{R}}}^2 \label{formu:criterion2}}
which has been studied in \cite{stoica2011new,stoica2011spice}. We defer the latter case to Subsection \ref{sec:TlessM}.

\begin{rem} The covariance fitting criteria above exploit the assumption that the sources are uncorrelated which results in the expression of $\m{R}$ in (\ref{formu:ycovariance}). However, a theoretical explanation is provided in \cite{stoica2011new,stoica2011spice} to show that the criteria are robust to correlations of the sources. Moreover, according to \cite{stoica2011spice,rojas2013note} they are connected to the $\ell_1$ norm minimization which is known to be robust to the correlations. Consequently, such robustness property is maintained in the method proposed in this paper which utilizes the same criteria.
\end{rem}

A simple calculation shows that
\equ{\begin{split}f_1
&=\tr\mbra{\m{R}^{-\frac{1}{2}}\sbra{\widetilde{\m{R}}-\m{R}} \widetilde{\m{R}}^{-1} \sbra{\widetilde{\m{R}}-\m{R}}\m{R}^{-\frac{1}{2}}}\\
&=\tr\mbra{\sbra{\m{R}^{-1}\widetilde{\m{R}}-\m{I}} \sbra{\m{I}-\widetilde{\m{R}}^{-1}\m{R}}}\\
&=\tr\sbra{\m{R}^{-1}\widetilde{\m{R}}}+ \tr\sbra{\widetilde{\m{R}}^{-1}\m{R}}-2M. \label{formu:f1_2}
\end{split}}
It is challenging to minimize $f_1$ with respect to the unknown parameters $\m{\theta}$, $\m{p}$ and $\m{\sigma}$ due to their nonlinear relation to $\m{R}$ by (\ref{formu:ycovariance}). We propose to estimate $\m{R}$ firstly by reparameterization and then determine the parameters of interest. Let \equ{\m{C}\sbra{\m{\theta},\m{p}}=\m{A}\sbra{\m{\theta}}\diag\sbra{\m{p}}\m{A}^H\sbra{\m{\theta}}. \label{formu:C} }
It is easy to see that $\m{C}\geq\m{0}$ and $\rank\sbra{\m{C}}=K\leq M-1$. Moreover,
\equ{C_{jl}=\sum_{k=1}^K p_ka_j\sbra{\theta_k}\overline{a}_l\sbra{\theta_k}=\sum_{k=1}^K p_ke^{i2\pi\sbra{j-l}\theta_k}.}
$\m{C}$ is thus a (Hermitian) Toeplitz matrix which is determined by $M$ complex numbers and can be written as $\m{C}=T\sbra{\m{u}}$ for some $\m{u}\in\bC^M$, where
\equ{T\sbra{\m{u}}=\begin{bmatrix}u_1 & u_2 & \cdots & u_M\\ \overline{u}_2 & u_1 & \cdots & u_{M-1}\\ \vdots & \vdots & \ddots & \vdots \\ \overline{u}_M & \overline{u}_{M-1} & \cdots & u_1\end{bmatrix}. \label{formu:Toeplitz}}

\subsection{SDP Formulations in the Case of Different $\lbra{\sigma_j}$}
It follows from (\ref{formu:Toeplitz}) that the characterization \equ{\m{R}\sbra{\m{u},\m{\sigma}}=T\sbra{\m{u}}+\diag\sbra{\m{\sigma}}\label{formu:R_in_u}}
captures the structure of $\m{R}$ under the constraints $T\sbra{\m{u}}\geq\m{0}$ and $\m{\sigma}\succeq\m{0}$. Further observations reveal that $\m{R}$ is inherently determined by $2M-1$ numbers ($M$ for the diagonal and the other $M-1$ for the off diagonal). However, $2M$ numbers ($M$ in $\m{u}$ and the other $M$ in $\m{\sigma}$) have been used in the above expression of $\m{R}$ ($M+1$ for the diagonal and $M-1$ for the off diagonal). As a result, redundancy exists along the diagonal of $\m{R}$. The effect of the redundancy is twofold. On one hand, it will cause an identifiability problem, which  will be tackled in Subsection \ref{sec:postproc}, that in general $\m{u}$, $\m{\sigma}$ cannot be uniquely identified from $\m{R}$ without accounting for additional information except $T\sbra{\m{u}}\geq\m{0}$ and $\m{\sigma}\succeq\m{0}$. It means that the solution $\m{\sigma}^*$ of the SDPs to be presented cannot be directly used as the final estimate of the noise variance. On the other hand, the adoption of one redundant variable enables us not to impose the nonconvex rank-deficiency constraint on $T\sbra{\m{u}}$ to characterize $\m{R}$ (it will be clarified in Subsection \ref{sec:postproc}). However, the redundancy problem will not affect the estimation of $\m{R}$ which is our current focus.

\subsubsection{The Case of $N\geq M$}
With the characterization of $\m{R}$ in (\ref{formu:R_in_u}) the minimization of $f_1$ is equivalent to
\equ{\begin{split}
&\min_{\m{u},\lbra{\m{\sigma}\succeq\m{0}}} \tr\sbra{\m{R}^{-1}\widetilde{\m{R}}}+ \tr\sbra{\widetilde{\m{R}}^{-1}\m{R}}, \\
&\st T\sbra{\m{u}}\geq\m{0}.\end{split} \label{formu:problem_Tlarge}}
Then we can show the following equivalences:
\equ{\begin{split}
\text{(\ref{formu:problem_Tlarge})}
\Leftrightarrow&\min_{\m{u},\lbra{\m{\sigma}\succeq\m{0}}} \tr\sbra{\widetilde{\m{R}}^{\frac{1}{2}}\m{R}^{-1}\widetilde{\m{R}}^{\frac{1}{2}}}+ \tr\sbra{\widetilde{\m{R}}^{-1}\m{R}}, \\
&\st T\sbra{\m{u}}\geq\m{0}\\
\Leftrightarrow& \min_{\m{X},\m{u},\lbra{\m{\sigma}\succeq\m{0}}} \tr\sbra{\m{X}}+ \tr\sbra{\widetilde{\m{R}}^{-1}\m{R}}, \\
&\st T\sbra{\m{u}}\geq\m{0} \text{ and } \m{X}\geq \widetilde{\m{R}}^{\frac{1}{2}}\m{R}^{-1}\widetilde{\m{R}}^{\frac{1}{2}}\\
\Leftrightarrow& \min_{\m{X},\m{u},\lbra{\m{\sigma}\succeq\m{0}}} \tr\sbra{\m{X}}+ \tr\sbra{\widetilde{\m{R}}^{-1}\m{R}}, \\
&\st \begin{bmatrix}\m{X}& \widetilde{\m{R}}^{\frac{1}{2}} & \\ \widetilde{\m{R}}^{\frac{1}{2}} & \m{R} & \\ & & T\sbra{\m{u}} \end{bmatrix}\geq\m{0}.\end{split} \label{formu:SDP_Nlarge}}
So the problem in (\ref{formu:problem_Tlarge}) can be formulated as an SDP and thus is convex. As the result, $\m{R}$ can be estimated by solving the SDP with its estimate given by $\widehat{\m{R}}=T\sbra{\m{u}^*}+\diag\sbra{\m{\sigma}^*}$, where $\sbra{\m{u}^*,\m{\sigma}^*}$ is the solution of the SDP.

\subsubsection{The Case of $N<M$} \label{sec:TlessM}
When $N<M$, $\widetilde{\m{R}}$ ceases to be nonsingular and the criterion in (\ref{formu:criterion2}) is used which can be written into
\equ{\begin{split}f_2\sbra{\m{u},\m{\sigma}}
&=\tr\mbra{\sbra{\widetilde{\m{R}}-\m{R}}\m{R}^{-1}\sbra{\widetilde{\m{R}}-\m{R}}}\\
&=\tr\sbra{\widetilde{\m{R}}\m{R}^{-1}\widetilde{\m{R}}}+\tr\sbra{\m{R}}-2\tr\sbra{\widetilde{\m{R}}}
\end{split} \label{formu:f2_2}}
Hence, an SDP similar to (\ref{formu:SDP_Nlarge}) can be formulated as follows:
\equ{\begin{split}
&\min_{\m{X},\m{u},\lbra{\m{\sigma}\succeq\m{0}}} \tr\sbra{\m{X}}+ \tr\sbra{\m{R}}, \\
&\st \begin{bmatrix}\m{X}& \widetilde{\m{R}} & \\ \widetilde{\m{R}} & \m{R} & \\ & & T\sbra{\m{u}} \end{bmatrix}\geq\m{0}.\end{split} \label{formu:SDP_Nsmall}}


\subsection{SDP Formulations in the Case of Equal $\lbra{\sigma_j}$}
It is reasonable to assume that the noise variances $\sigma_j$, $j\in\mbra{M}$, are equal in some scenarios. Though the formulations in the last subsection can be applied to such a case by imposing the constraint $\sigma_1=\dots=\sigma_M=\sigma$, simpler formulations in fact exist. In such a case, $\diag\sbra{\m{\sigma}}=\sigma\m{I}$ with $\sigma\in\bR_+$. Therefore, the covariance $\m{R}$ has a Toeplitz structure and the characterization
\equ{\m{R}=T\sbra{\m{u}} \label{formu:R_equalvar}}
captures its structure without redundancy under the constraint $T\sbra{\m{u}}\geq\m{0}$ for some $\m{u}\in\bC^M$. Note that the $T\sbra{\m{u}}$ here is the same as that in (\ref{formu:R_in_u}) off the diagonal but different on the diagonal, i.e., the two $\m{u}$'s are the same except the first entry. Following from the similar procedures as in the last subsection, similar SDPs can be formulated. In particular, when $N\geq M$ we have
\equ{\begin{split}
&\min_{\m{X},\m{u}} \tr\sbra{\m{X}}+ \tr\sbra{\widetilde{\m{R}}^{-1}T\sbra{\m{u}}}, \\
&\st \begin{bmatrix}\m{X}& \widetilde{\m{R}}^{\frac{1}{2}} \\ \widetilde{\m{R}}^{\frac{1}{2}} & T\sbra{\m{u}} \end{bmatrix}\geq\m{0}.\end{split} \label{formu:SDP_Nlarge_identical}}
When $N<M$,
\equ{\begin{split}
&\min_{\m{X},\m{u}} \tr\sbra{\m{X}}+ \tr\sbra{T\sbra{\m{u}}}, \\
&\st \begin{bmatrix}\m{X}& \widetilde{\m{R}} \\ \widetilde{\m{R}} & T\sbra{\m{u}} \end{bmatrix}\geq\m{0}.\end{split} \label{formu:SDP_Nsmall_identical}}
Note that the constraint $T\sbra{\m{u}}\geq\m{0}$ has been implicitly included in the constraints in (\ref{formu:SDP_Nlarge_identical}) and (\ref{formu:SDP_Nsmall_identical}). The simplified SDP formulations will lead to faster computations in practice. By solving one of the SDPs, $\widehat{\m{R}}$ is obtained as $\widehat{\m{R}}=T\sbra{\m{u}^*}$ given the solution $\m{u}^*$. For notational convenience, we say $\m{\sigma}^*=\m{0}$ in this case. In fact, note that (\ref{formu:R_equalvar}), (\ref{formu:SDP_Nlarge_identical}) and (\ref{formu:SDP_Nsmall_identical}) can be obtained from (\ref{formu:R_in_u}), (\ref{formu:SDP_Nlarge}) and (\ref{formu:SDP_Nsmall}), respectively, by setting $\m{\sigma}=\m{0}$.

\subsection{Postprocessing of $\widehat{\m{R}}$} \label{sec:postproc}
After obtaining $\widehat{\m{R}}$, the following task is to estimate the parameters $\m{\theta}$, $\m{p}$ and $\m{\sigma}$ by writing it back into the form of (\ref{formu:ycovariance}). To do this, we decompose $\widehat{\m{R}}$ into
\equ{\widehat{\m{R}}=T\sbra{\widehat{\m{u}}}+\diag\sbra{\widehat{\m{\sigma}}}, \label{formu:postproc}}
where $T\sbra{\widehat{\m{u}}}=\m{A}\sbra{\widehat{\m{\theta}}}\diag\sbra{\widehat{\m{p}}} \m{A}^H\sbra{\widehat{\m{\theta}}}\geq\m{0}$ is the estimate of $\m{C}\sbra{\m{\theta},\m{p}}$ in (\ref{formu:C}) and $\widehat{\m{\sigma}}\succeq\m{0}$ is the noise covariance estimate. Such $\sbra{\widehat{\m{u}},\widehat{\m{\sigma}}}$ always exists according to the way $\widehat{\m{R}}$ is given, however, it is generally not unique. In particular, for any $\delta\geq-\min\sbra{\m{\sigma}^*}$ satisfying that $T\sbra{\m{u}^*}-\delta\m{I}\geq\m{0}$, $\sbra{\widehat{\m{u}},\widehat{\m{\sigma}}}= \sbra{\m{u}^*-\begin{bmatrix}\delta\\\m{0}\end{bmatrix}, \m{\sigma}^*+\delta\m{1}}$ leads to one realization of the decomposition, which in fact also enumerates all possible realizations. We utilize the prior knowledge that $K\leq M-1$ to make the decomposition unique as follows. It follows from $K\leq M-1$ that $\rank\sbra{\m{C}}=K\leq M-1$. Therefore, it is natural to impose that $\rank\sbra{T\sbra{\widehat{\m{u}}}}\leq M-1$, i.e., $T\sbra{\widehat{\m{u}}}=T\sbra{\m{u}^*}-\delta\m{I}$ is rank-deficient. A direct result is that $\delta$ is an eigenvalue of $T\sbra{\m{u}^*}$. Then by $T\sbra{\m{u}^*}-\delta\m{I}\geq\m{0}$ we see that
\equ{\delta=\lambda_{\text{min}}\sbra{T\sbra{\m{u}^*}}}
and thus the decomposition is unique, where $\lambda_{\text{min}}\sbra{\cdot}$ denotes the minimum eigenvalue.\footnote{The postprocessing as in (\ref{formu:postproc}) can be done if we have only $\widehat{\m{R}}$ rather than $\m{u}^*$ and $\m{\sigma}^*$. To see this, let $\widetilde{\m{u}}$ be the transpose of the first row of $\widehat{\m{R}}$. Then $\widehat{\m{u}}=\widetilde{\m{u}}-\begin{bmatrix}\lambda_{\text{min}}\sbra{T\sbra{\widetilde{\m{u}}}}\\\m{0}\end{bmatrix}$ and $\widehat{\m{\sigma}}$ is obtained as the diagonal of $\widehat{\m{R}}$ minus $\widehat{u}_1\m{1}$.}

\begin{rem} The postprocessing is used to separate sources and noise in the estimated covariance matrix such that the source part can be represented by as few sources as possible based on the minimum description length principle \cite{grunwald2007minimum} (see the next subsection for clarity). In fact, the idea of postprocessing has been studied in the literature for a special case where $\widehat{\m{R}}$ is a Toeplitz matrix as in the case of equal $\lbra{\sigma_j}$ (see, e.g., \cite[Section 4.9.2]{stoica2005spectral}). It is noted that the postprocessing is very important in SPA, without which the final parameter estimate is generally not unique and does not have the statistical properties that will be shown in Section \ref{sec:property}. Even worse, the frequency estimate of SPA can cease to be sparse. To see this, suppose that $T\sbra{\m{u}^*}$ has full rank (this is generally the case in the presence of noise). Choose $\theta_1\in\left[0,1\right)$ arbitrarily and let $p_1<\sbra{\m{a}^H\sbra{\theta_1} T^{-1}\sbra{\m{u}^*}\m{a}\sbra{\theta_1}}^{-1}$. It follows that the residue $T\sbra{\m{u}^*}- p_1\m{a}\sbra{\theta_1}\m{a}^H\sbra{\theta_1}>\m{0}$ which still has the Toeplitz structure. Then we can choose $\theta_2$ and $p_2$ similarly based on the residue. The process can be repeated infinitely many times and results in infinitely long vectors $\m{\theta}$ and $\m{p}$. \label{rem:postproc}
\end{rem}

\subsection{Frequency and Power Solutions} \label{sec:freqpowest}
The remaining task is to retrieve the frequency estimate $\widehat{\m{\theta}}$ and the power estimate $\widehat{\m{p}}$ given $T\sbra{\widehat{\m{u}}}$, which is based on the following classical Vandermonde decomposition lemma for positive semidefinite Toeplitz matrices (see, e.g., \cite{grenander1958toeplitz,stoica2005spectral}).
\begin{lem} Any positive semidefinite Toeplitz matrix $T\sbra{\m{u}}\in\bC^{M\times M}$ can be represented as
\equ{T\sbra{\m{u}}=\m{V}\m{P}\m{V}^H, \label{formu:toeprep}}
where
{\lentwo\equa{\m{V}
&=&\mbra{\m{a}\sbra{\theta_1},\dots,\m{a}\sbra{\theta_r}}, \\ \m{P}
&=&\diag\sbra{p_1,\dots,p_r},
}}$\theta_j\in\left[0,1\right)$, $p_j>0$ for $j\in[r]$, and $r=\rank\sbra{T\sbra{\m{u}}}$. Moreover, the representation is unique up to permutation of elements of $\m{\theta}$ and $\m{p}$ if $r\leq M-1$. $\qed$ \label{lem:toeplitz}
\end{lem}

It follows from Lemma \ref{lem:toeplitz} that $\widehat{\m{\theta}}$ and $\widehat{\m{p}}$ can be uniquely determined given $T\sbra{\widehat{\m{u}}}$ since $\rank\sbra{T\sbra{\widehat{\m{u}}}}\leq M-1$. In practice, $\widehat{\m{\theta}}$ and $\widehat{\m{p}}$ can be obtained as follows. Given $T\sbra{\widehat{\m{u}}}=\m{A}\sbra{\widehat{\m{\theta}}}\diag\sbra{\widehat{\m{p}}} \m{A}^H\sbra{\widehat{\m{\theta}}}$, it is easy to show that
\equ{\begin{bmatrix}\m{A}\sbra{\widehat{\m{\theta}}}\\ \overline{\m{A}}_{\lbra{2,\dots,M}}\sbra{\widehat{\m{\theta}}}\end{bmatrix}\widehat{\m{p}}= \begin{bmatrix}\overline{\widehat{\m{u}}} \\ \widehat{\m{u}}_{\lbra{2,\dots,M}}\end{bmatrix} \label{formu:linsys}}
since $\widehat{\m{p}}\succ\m{0}$, where $\overline{\m{A}}_{\lbra{2,\dots,M}}\sbra{\widehat{\m{\theta}}}$ takes all but the first rows of the matrix $\overline{\m{A}}\sbra{\widehat{\m{\theta}}}$ which denotes the complex conjugate of  $\m{A}\sbra{\widehat{\m{\theta}}}$. So we build a system of $2M-1$ equations that is linear in $\widehat{\m{p}}$ whose length is maximally $M-1$, where each column of the coefficient matrix corresponds to a uniformly sampled sinusoid (after permutation of rows). According to \cite{blu2008sparse}, Prony's method can be applied to this type of systems to efficiently solve $\widehat{\m{\theta}}$ and $\widehat{\m{p}}$. In particular, $\widehat{\m{\theta}}$ is firstly obtained from zeros of a polynomial which is obtained by solving a linear system involving only $\widehat{\m{u}}$. After that, $\widehat{\m{p}}$ is solved from (\ref{formu:linsys}). Readers are referred to \cite{blu2008sparse} for the detailed procedure.

The proposed SPA algorithm for the array processing with a ULA is presented in Algorithm \ref{alg:SPA}. The covariance matrix $\m{R}$ is firstly estimated by solving an SDP. Then the postprocessing procedure is carried out for $\widehat{\m{R}}$ to resolve an identifiability problem that exists in the solution of the SDP (the noise variance is estimated at this step). Finally, the frequency and the power are solved using Prony's method. To carry out the parameter estimation, SPA requires only the data snapshots $\m{Y}$ of the sensor array without any other user parameters. More properties of SPA will be presented in Section \ref{sec:property}.

\begin{algorithm}
\caption{Sparse and parametric approach (SPA)}
Input: observed snapshots $\m{Y}$.
\begin{enumerate}
 \item Estimate $\m{R}$ by solving the solution $\sbra{\m{u}^*,\m{\sigma}^*}$ of an SDP;
 \item Postprocess $\widehat{\m{R}}$ to obtain $\widehat{\m{u}}$ and $\widehat{\m{\sigma}}$ via $T\sbra{\widehat{\m{u}}}=T\sbra{\m{u}^*}-\lambda_{\text{min}}\sbra{T\sbra{\m{u}^*}}\m{I}$ and $\widehat{\m{\sigma}} =\m{\sigma}^*+\lambda_{\text{min}}\sbra{T\sbra{\m{u}^*}}\m{1}$;
 \item Solve $\widehat{\m{\theta}}$ and $\widehat{\m{p}}$ from (\ref{formu:linsys}) using Prony's method.
\end{enumerate}
Output: parameter estimator $\sbra{\widehat{\m{\theta}}, \widehat{\m{p}}, \widehat{\m{\sigma}}}$.
\label{alg:SPA}
\end{algorithm}

\begin{rem}
In general, SPA cannot correctly determine the true number of sources $K$ but provides additionally spurious sources. This is because that we do not assume in SPA any knowledge of the source number or the noise variance(s). Consequently, SPA solves the problem as in the worst case where up to $M-1$ sources can be present. In fact, the source number estimation problem itself is very difficult which is known as model order selection \cite{stoica2004model}, and lacking guarantees on the estimated source number seems to be a common feature of sparse methods, especially when no prior knowledge mentioned above is available (see, e.g., \cite{malioutov2005sparse,stoica2011spice,stoica2012spice,carlin2013directions}). A positive side of SPA is that it has strong statistical properties which will be shown in Section \ref{sec:property}, and theoretically the powers of the spurious sources can be very close to zero under sufficient snapshots or appropriately low SNR, which enables the spurious sources to be distinguished from the real ones. Currently, methods have been proposed to remove spurious sources by modifying the solution of a sparse method. For example, intuitive thresholding is used in \cite{carlin2013directions} and information criteria-based methods are presented in \cite{yardibi2010source,austin2010relation}. In future studies, we may seek to incorporate model order selection in SPA such that the source number can be estimated automatically together with the parameters but it is beyond the scope of this paper.
\end{rem}

\section{SPA with an SLA} \label{sec:SPA_SLA}



In this section we extend SPA to the SLA case. With respect to an SLA $\m{\Omega}$, denote the steering matrix by $\m{A}_{\m{\Omega}}\sbra{\m{\theta}}=\mbra{\m{a}_{\m{\Omega}}\sbra{\theta_1}, \dots, \m{a}_{\m{\Omega}}\sbra{\theta_K}}\in\bC^{L\times K}$, the data snapshots by $\m{Y}_{\m{\Omega}}=\mbra{\m{y}_{\m{\Omega}}\sbra{1}, \dots, \m{y}_{\m{\Omega}}\sbra{N}}=\m{A}_{\m{\Omega}}\sbra{\m{\theta}}\m{S}+\m{E}_{\m{\Omega}}\in\bC^{L\times N}$, the covariance matrix by $\m{R}_{\m{\Omega}}=E\lbra{\m{y}_{\m{\Omega}}(t)\m{y}^H_{\m{\Omega}}(t)}=\m{A}_{\m{\Omega}}\sbra{\m{\theta}} \diag\sbra{\m{p}}\m{A}^H_{\m{\Omega}}\sbra{\m{\theta}} +\diag\sbra{\m{\sigma}_{\m{\Omega}}}\in\bC^{L\times L}$ and the sample covariance by $\widetilde{\m{R}}_{\m{\Omega}}=\frac{1}{N}\m{Y}_{\m{\Omega}}\m{Y}^H_{\m{\Omega}}\in\bC^{L\times L}$, where $\m{\sigma}_{\m{\Omega}}\in\bR_+^L$ denotes the noise variance parameter. It follows from (\ref{formu:steervector_SLA}) that $\m{A}_{\m{\Omega}}\sbra{\m{\theta}}=\m{\Gamma}_{\m{\Omega}}\m{A}\sbra{\m{\theta}}$ and then $\m{A}_{\m{\Omega}}\sbra{\m{\theta}}\diag\sbra{\m{p}} \m{A}_{\m{\Omega}}^H\sbra{\m{\theta}} =\m{\Gamma}_{\m{\Omega}}T\sbra{\m{u}}\m{\Gamma}_{\m{\Omega}}^T\triangleq T_{\m{\Omega}}\sbra{\m{u}} $, where the Toeplitz matrix $T\sbra{\m{u}}=\m{A}\sbra{\m{\theta}}\diag\sbra{\m{p}}\m{A}^H\sbra{\m{\theta}}$ for some $\m{u}\in\bC^M$. A careful study of $T_{\m{\Omega}}\sbra{\m{u}}$ reveals that
\equ{T_{\m{\Omega}}\sbra{\m{u}}_{jl}=\sum_{k=1}^K p_ka_{\Omega_j}\sbra{\theta_k}\overline{a}_{\Omega_l}\sbra{\theta_k}=\sum_{k=1}^K p_ke^{i2\pi\sbra{\Omega_j-\Omega_l}\theta_k}, \label{formu:entry_SLA}}
i.e., the entries of $T_{\m{\Omega}}\sbra{\m{u}}$ are specified by the coarray $\m{\cD}$ defined in (\ref{formu:coarray}).

In this paper we are mainly interested in the case where the SLA $\m{\Omega}$ is a redundancy array. As in the ULA case, $K\leq M-1$ is considered as {\em a priori} knowledge according to Subsection \ref{sec:lineararray}. The matrix $T_{\m{\Omega}}\sbra{\m{u}}$ contains all elements of $\m{u}$ explicitly by (\ref{formu:entry_SLA}). It implies that the relation $T\sbra{\m{u}}\leftrightarrow T_{\m{\Omega}}\sbra{\m{u}}$ is one-to-one.


\begin{exa} Given a redundancy SLA $\m{\Omega}=\lbra{1,2,5,7}$, we have $\m{u}\in\bC^7$ and $T_{\m{\Omega}}\sbra{\m{u}}=\m{\Gamma}_{\m{\Omega}}T\sbra{\m{u}}\m{\Gamma}_{\m{\Omega}}^T=\begin{bmatrix}u_1 & u_2 & u_5 & u_7\\ \overline{u}_2 & u_1 & u_4 & u_6 \\ \overline{u}_5 & \overline{u}_4 & u_1 & u_3 \\ \overline{u}_7 & \overline{u}_6 & \overline{u}_3 & u_1\end{bmatrix}$, where all elements of $\m{u}$ are contained.
\end{exa}

It follows that the covariance matrix $\m{R}_{\m{\Omega}}$ can be characterized as
\equ{\begin{split}\m{R}_{\m{\Omega}}
&=T_{\m{\Omega}}\sbra{\m{u}} +\diag\sbra{\m{\sigma}_{\m{\Omega}}}\\
&=\m{\Gamma}_{\m{\Omega}}T\sbra{\m{u}}\m{\Gamma}_{\m{\Omega}}^T +\diag\sbra{\m{\sigma}_{\m{\Omega}}}\end{split}}
under the constraints $T\sbra{\m{u}}\geq\m{0}$ and $\m{\sigma}_{\m{\Omega}}\succeq\m{0}$. Consider similar covariance fitting criteria as in the ULA case, i.e., the minimization of $\frobn{\m{R}_{\m{\Omega}}^{-\frac{1}{2}} \sbra{\widetilde{\m{R}}_{\m{\Omega}}-\m{R}_{\m{\Omega}}}\widetilde{\m{R}}_{\m{\Omega}}^{-\frac{1}{2}}}^2$ when $N\geq L$ and $\frobn{\m{R}_{\m{\Omega}}^{-\frac{1}{2}} \sbra{\widetilde{\m{R}}_{\m{\Omega}}-\m{R}_{\m{\Omega}}}}^2$ otherwise. Then SPA can be extended to the SLA case. In particular, when $N\geq L$ we obtain the following SDP:
\equ{\begin{split}&\min_{\m{X},\m{u},\lbra{\m{\sigma}_{\m{\Omega}}\succeq\m{0}}} \tr\sbra{\m{X}}+ \tr\sbra{\widetilde{\m{R}}_{\m{\Omega}}^{-1}\m{R}_{\m{\Omega}}}, \\
&\st \begin{bmatrix}\m{X}& \widetilde{\m{R}}_{\m{\Omega}}^{\frac{1}{2}} & \\ \widetilde{\m{R}}_{\m{\Omega}}^{\frac{1}{2}} & \m{R}_{\m{\Omega}} & \\ & & T\sbra{\m{u}} \end{bmatrix}\geq\m{0}.\end{split} \label{formu:SDP_Nlarge_miss}}
When $N<L$, it is
\equ{\begin{split}
&\min_{\m{X},\m{u},\lbra{\m{\sigma}_{\m{\Omega}}\succeq\m{0}}} \tr\sbra{\m{X}}+ \tr\sbra{\m{R}_{\m{\Omega}}}, \\
&\st \begin{bmatrix}\m{X}& \widetilde{\m{R}}_{\m{\Omega}} & \\ \widetilde{\m{R}}_{\m{\Omega}} & \m{R}_{\m{\Omega}} & \\ & & T\sbra{\m{u}} \end{bmatrix}\geq\m{0}.\end{split} \label{formu:SDP_Nsmall_miss}}


In the case of equal noise variances, $\m{R}_{\m{\Omega}}$ can be characterized as $\m{R}_{\m{\Omega}}=\m{\Gamma}_{\m{\Omega}}T\sbra{\m{u}}\m{\Gamma}_{\m{\Omega}}^T$. Then similar SDPs can be formulated as in (\ref{formu:SDP_Nlarge_miss}) and (\ref{formu:SDP_Nsmall_miss}) by simply setting $\m{\sigma}_{\m{\Omega}}=\m{0}$. Note that the problem dimension in such a case cannot be further reduced as in the ULA case since $\m{\Gamma}_{\m{\Omega}}T\sbra{\m{u}}\m{\Gamma}_{\m{\Omega}}^T\geq\m{0}$ does not imply $T\sbra{\m{u}}\geq\m{0}$.

Note that the same identifiability problem exists in the solution $\sbra{\m{u}^*,\m{\sigma}_{\m{\Omega}}^*}$ of the SDPs above. To resolve the problem, the same procedures can be carried out as in the ULA case by exploiting that $\rank\sbra{T\sbra{\widehat{\m{u}}}}\leq M-1$. Then the parameter estimate $\sbra{\widehat{\m{\theta}}, \widehat{\m{p}}, \widehat{\m{\sigma}}_{\m{\Omega}}}$ can be obtained.

\begin{rem}
When the SLA $\m{\Omega}$ is not a redundancy array, i.e., the coarray $\m{\cD}$ is not a ULA, SPA presented above can be applied straightforwardly. In such a case, it is natural to require that $K\leq \overline{M}-1<M-1$, where $\overline{M}-1$ denotes the maximum number of sources detectable using the non-redundancy SLA. However, the introduced implementation of SPA can only use the information up to $K\leq M-1$, which should be a waste of knowledge. A thorough study on exploitation of the full information should be investigated in the future but is beyond the scope of this paper.
\end{rem}

\section{Properties of SPA} \label{sec:property}
\subsection{Sparse and Parametric Method}
As its name suggests, SPA is a {\em sparse} and {\em parametric} method. It carries out the optimization by reparameterization and thus is a parametric method. The length of its frequency estimator $\widehat{\m{\theta}}$ is maximally $M-1$. Since this is an important result of this paper, we state it formally in the following theorem.
\begin{thm} The length of the frequency estimator $\widehat{\m{\theta}}$ of SPA is maximally $M-1$. $\qed$
\end{thm}

\subsection{Statistical Properties} \label{sec:statproperty}
We consider ULAs and redundancy SLAs in this subsection, or collectively, redundancy arrays. Notice that dimensions of the parameter estimators $\widehat{\m{\theta}}$ and $\widehat{\m{p}}$ may be different from the true dimension $K$. To discuss statistical properties of the SPA estimator, the dimension problem should be resolved first. Notice that the proposed SPA is equivalently assuming $K=M-1$ (the worst case) with the knowledge $K\leq M-1$. Then we expand both the true parameter and its estimator to the same dimension. Define
\lentwo{\equa{\cE_p\sbra{\m{v}}
&=&\begin{bmatrix}\m{v}\\ \m{0}\end{bmatrix}\in\bR^{M-1},\\ \cE_f\sbra{\m{v}}
&=&\begin{bmatrix}\m{v}\\ \m{w}\end{bmatrix} \in\bR^{M-1}}
}for a vector $\m{v}$ of dimension no more than $M-1$, where $\m{w}\in\left[0,1\right)^{M-1-\abs{\m{v}}}$ is arbitrary. It is easy to see that $\sbra{\cE_f\sbra{\m{\theta}}, \cE_p\sbra{\m{p}}}$ and $\sbra{\m{\theta}, \m{p}}$ are physically equivalent since all of the added virtual sources have zero powers. Then we have the following results.

\begin{thm} Assume that $\lbra{\m{s}(t)}$ and $\lbra{\m{e}(t)}$ are uncorrelated, $E\mbra{\m{s}(t_1)\m{s}^H(t_2)}=\diag\sbra{\m{p}}\delta_{t_1,t_2}$ and $E\mbra{\m{e}(t_1)\m{e}^H(t_2)}=\diag\sbra{\m{\sigma}}\delta_{t_1,t_2}$ for $t_1,t_2\in\mbra{N}$. Moreover, assume that the sensor array is a redundancy array and $K\leq M-1$. Then the parameter estimator $\sbra{\widehat{\m{\theta}}, \widehat{\m{p}}, \widehat{\m{\sigma}}}$ of SPA is statistically consistent (in $N$).\footnote{Without ambiguity, $\widehat{\m{\sigma}}$ and $\m{\sigma}$ denote respectively $\widehat{\m{\sigma}}_{\m{\Omega}}$ and $\m{\sigma}_{\m{\Omega}}$ in the SLA case, or $\widehat{\sigma}$ and $\sigma$ in the case of equal noise variances. It is the same for Theorems \ref{thm:AML} and \ref{thm:efficiency}.} \label{thm:consistency}
\end{thm}
\begin{proof} Without loss of generality, we consider only the ULA case. The proof is based on the observation that the parameter estimator $\sbra{\widehat{\m{\theta}}, \widehat{\m{p}}, \widehat{\m{\sigma}}}$ can be uniquely determined given the covariance estimator $\widehat{\m{R}}$. Suppose that the true parameter value is $\sbra{\m{\theta}^o, \m{p}^o, \m{\sigma}^o}$. As $N\rightarrow\infty$, $\widetilde{\m{R}}$ approaches the true covariance $\m{R}^o=\m{A}\sbra{\m{\theta}^o}\diag\sbra{\m{p}^o}\m{A}^H\sbra{\m{\theta}^o} + \diag\sbra{\m{\sigma}^o}$. Then by (\ref{formu:criterion1}) the SDP of SPA admits a unique minimizer $\m{R}^o$ of $\m{R}$, which determines the unique parameter estimate $\sbra{\m{\theta}^o, \m{p}^o, \m{\sigma}^o}$ given $K\leq M-1$.
\end{proof}

\begin{thm} Assume that $\lbra{\m{s}(t)}$ and $\lbra{\m{e}(t)}$ are uncorrelated and both are i.i.d.~circular Gaussian with means zero and covariance matrices $\diag\sbra{\m{p}}$ and $\diag\sbra{\m{\sigma}}$, respectively. Moreover, assume that the sensor array is a redundancy array and $K\leq M-1$. Then $\sbra{\cE_f\sbra{\widehat{\m{\theta}}}, \cE_p\sbra{\widehat{\m{p}}}, \widehat{\m{\sigma}}}$ is asymptotically an ML estimator of $\sbra{\cE_f\sbra{\m{\theta}}, \cE_p\sbra{\m{p}}, \m{\sigma}}$. \label{thm:AML}
\end{thm}
\begin{proof} Consider first the ULA case. Since $\sbra{\widehat{\m{\theta}}, \widehat{\m{p}}, \widehat{\m{\sigma}}}$ can be uniquely determined given $\widehat{\m{R}}$, the SPA estimator is equivalent to solving the following (nonconvex) optimization problem:\footnote{The zero entries, if any, of the solution of $\m{p}$ are removed in SPA as well as the corresponding entries of the solution of $\m{\theta}$.}
\equ{\begin{split}
&\min_{\m{\theta}\in\left[0,1\right)^{M-1},\m{p}\in\bR_+^{M-1},\m{\sigma}\succeq\m{0}} \frobn{\m{R}^{-\frac{1}{2}}\sbra{\widetilde{\m{R}}-\m{R}}\widetilde{\m{R}}^{-\frac{1}{2}}}^2,\\
&\st \m{R}=\m{A}\sbra{\m{\theta}}\diag\sbra{\m{p}}\m{A}^H\sbra{\m{\theta}}+ \diag\sbra{\m{\sigma}}. \end{split} \label{formu:SPA_nonconvex}}
On the other hand, the data snapshots $\lbra{\m{y}(t)}$ are i.i.d. Gaussian with mean zero and covariance $\m{R}$ under the assumptions. According to the extended invariance principle (EXIP) \cite{stoica1989reparametrization} and the derivations in \cite{ottersten1998covariance}, the global minimizer of (\ref{formu:SPA_nonconvex}) is a large-snapshot realization of the ML estimator of $\sbra{\cE_f\sbra{\m{\theta}}, \cE_p\sbra{\m{p}}, \m{\sigma}}$.

Similarly, the SPA estimator in the redundancy SLA case is the global minimizer of the following problem:
\equ{\begin{split}
&\min_{\m{\theta}\in\left[0,1\right)^{M-1},\m{p}\in\bR_+^{M-1},\m{\sigma}_{\m{\Omega}}\succeq\m{0}} \frobn{\m{R}_{\m{\Omega}}^{-\frac{1}{2}} \sbra{\widetilde{\m{R}}_{\m{\Omega}}-\m{R}_{\m{\Omega}}}\widetilde{\m{R}}_{\m{\Omega}}^{-\frac{1}{2}}}^2,\\
&\st \m{R}_{\m{\Omega}}=\m{A}_{\m{\Omega}}\sbra{\m{\theta}}\diag\sbra{\m{p}}\m{A}_{\m{\Omega}}^H\sbra{\m{\theta}}+ \diag\sbra{\m{\sigma}_{\m{\Omega}}}. \end{split} \label{formu:SPA_nonconvex_SLA}}
Then the same result follows.
\end{proof}

Theorem \ref{thm:AML} states that $\sbra{\cE_f\sbra{\widehat{\m{\theta}}}, \cE_p\sbra{\widehat{\m{p}}}, \widehat{\m{\sigma}}}$ is asymptotically an ML estimator of $\sbra{\cE_f\sbra{\m{\theta}}, \cE_p\sbra{\m{p}}, \m{\sigma}}$ under some technical assumptions. Theorem \ref{thm:consistency} shows that this estimator is also consistent. Since, asymptotically, the power estimator $\cE_p\sbra{\widehat{\m{p}}}$ may lie on the boundary of $\bR_+^{M-1}$ and thus the asymptotic normality of $\sbra{\cE_f\sbra{\widehat{\m{\theta}}}, \cE_p\sbra{\widehat{\m{p}}}, \widehat{\m{\sigma}}}$ does not hold directly. However, it indeed holds in the case of $K=M-1$ where the true parameter $\m{p}^o$ is an interior point of $\bR_+^{M-1}$. So we have the following theorem.

\begin{thm} Under the assumptions of Theorem \ref{thm:AML} and $K=M-1$ with $\m{p}\succ\m{0}$ and $\m{\sigma}\succ\m{0}$, the SPA estimator $\sbra{\widehat{\m{\theta}}, \widehat{\m{p}}, \widehat{\m{\sigma}}}$ is asymptotically normal with asymptotic unbiased mean $\sbra{\m{\theta}^o, \m{p}^o, \m{\sigma}^o}$ and covariance $\m{F}^{-1}$, where $\m{F}$ denotes the Fisher information matrix of the parameter. It follows that $\sbra{\widehat{\m{\theta}}, \widehat{\m{p}}, \widehat{\m{\sigma}}}$ is statistically asymptotically efficient. \label{thm:efficiency}
\end{thm}
\begin{proof} The result follows from Theorems \ref{thm:consistency} and \ref{thm:AML} and the properties of consistent ML estimators.
\end{proof}

\begin{rem} The asymptotic covariance matrix $\m{F}^{-1}$ is usually referred to as the Crammer-Rao lower bound (CRLB), which can be computed following from \cite{stoica1990performance} and will be omitted in this paper. Alternatively, its numerical results will be presented in Section \ref{sec:simulation}.
\end{rem}


\section{Connection to Prior Arts} \label{sec:connectionPrior}
\subsection{Connection to SPICE} \label{sec:connectionSPICE}
The SPA method proposed in this paper is closely connected to SPICE in \cite{stoica2011spice}. The two methods adopt the same covariance fitting criteria. Roughly speaking, SPICE can be considered as a discretized version of SPA when applied to ULAs and SLAs. However, the following two main differences make their performances very different. One is that SPA is a parametric method while SPICE is a semiparametric method. In particular, the covariance $\m{R}$ is approximated in SPICE by discretizing the continuous range $\left[0,1\right)$ of the frequency. Consider the ULA case as an example. Denote $\widetilde{\m{\theta}}\in\left[0,1\right)^{\widetilde{N}}$ the discretized sampling grid of the frequency and $\widetilde{\m{p}}\in\bR_+^{\widetilde{N}}$ the corresponding power vector, where $\widetilde{N}$ denotes the grid size. Then $\m{R}$ is expressed as
\equ{\m{R}\sbra{\widetilde{\m{p}}, \m{\sigma}}=\m{A}\sbra{\widetilde{\m{\theta}}}\diag\sbra{\widetilde{\m{p}}} \m{A}^H\sbra{\widetilde{\m{\theta}}}+\diag\sbra{\m{\sigma}} \label{formu:Rgrid}}
which is a linear function of $\sbra{\widetilde{\m{p}}, \m{\sigma}}$. With (\ref{formu:Rgrid}), the covariance fitting criterion (\ref{formu:criterion1}) or (\ref{formu:criterion2}) becomes a convex function of $\sbra{\widetilde{\m{p}}, \m{\sigma}}$ and is optimized in \cite{stoica2011new,stoica2011spice} via an iterative algorithm, named as SPICE.
Note that the characterization of $\m{R}$ in (\ref{formu:Rgrid}) is only an approximation since there is no guarantee that the true frequencies lie on the grid $\widetilde{\m{\theta}}$. Thus the approximation error (or modeling error), which depends on the grid density, is one potential reason causing estimation inaccuracy of SPICE. The other difference is that the parameter estimate of SPICE is obtained directly from the solution $\sbra{\widetilde{\m{p}}^*, \m{\sigma}^*}$ of the covariance fitting optimization problem while an additional postprocessing procedure is carried out in SPA. In particular, a source at $\widetilde{\theta}_j$ is detected in SPICE once $\widetilde{p}_j^*>0$ and the frequency estimate is constrained on the grid, which will be referred to as the on-grid issue hereafter and becomes a second potential reason causing inaccuracy of SPICE. Furthermore, the parameter estimate without the postprocessing is generally not unique according to Remark \ref{rem:postproc} due to an identifiability problem which also exists in the solution $\sbra{\widetilde{\m{p}}^*, \m{\sigma}^*}$ of SPICE. This can be a third potential reason causing inaccuracy of SPICE. In fact, the SPICE estimate might not be {\em sparse} since the support of $\widetilde{\m{p}}^*$ can be as large as its dimension $\widetilde{N}$ according to Remark \ref{rem:postproc}, which will be numerically verified in Section \ref{sec:simulation}.

In summary, there are three potential reasons that may cause inaccuracy to the parameter estimation of SPICE, including the modeling error, the on-grid issue and the identifiability problem, rendering that SPICE does not possess the sparse and the statistical properties of SPA presented in Section \ref{sec:property}. The first reason corresponds to the first difference and is introduced by the discretization, which can be alleviated by adopting a dense sampling grid but at the cost of more expensive computations. The last two are related to the second difference. Inspired by SPA, we present SPICE-PP to resolve them using the postprocessing technique in the next subsection.\footnote{Note that the on-grid issue is not necessarily a consequence of the discretization but can be solved using the postprocessing.}

Before proceeding to SPICE-PP, we compare computational costs of SPICE and SPA. We consider only the dominant part in the following comparison, i.e., the computation of $\widehat{\m{R}}^{\frac{1}{2}}$ or $\widehat{\m{R}}_{\m{\Omega}}^{\frac{1}{2}}$ is excluded for both SPICE and SPA which takes $O\sbra{M^2N+M^3}$ or $O\sbra{L^2N+L^3}$ flops, respectively, and the postprocessing and parameter solving are also excluded for SPA which take $O\sbra{M^3}$ flops. SPICE is an iterative algorithm whose computational complexity equals the complexity per iteration times the number of iterations, i.e., $O\sbra{M^3\widetilde{N}T}$ in the ULA case and $O\sbra{L^3\widetilde{N}T}$ in the SLA case, where $T$ denotes the number of iterations which is hard to quantify and empirically observed to vary in different scenarios. For SPA we adopt an off-the-shelf SDP solver, SDPT3 \cite{toh1999sdpt3}, where the interior-point method is implemented to solve the SDP. Denote by $n_1$ and $n_2\times n_2$ the variable size and dimension of the positive semidefinite matrix in the semidefinite constraint of an SDP, respectively. Then the SDP can be solved in $O\sbra{n_1^2n_2^{2.5}}$ flops in the worst case according to \cite{vandenberghe1996semidefinite}. In the ULA case, $n_1$ is on the order of $M^2$ and $n_2$ is proportional to $M$. It follows that the complexity of SPA is $O\sbra{M^{6.5}}$. In the SLA case, $n_1$ and $n_2$ are on the order of $L^2+M$ and $L+M$, respectively. Then the complexity is $O\sbra{L^4M^{2.5}+M^{4.5}}$ since $L\leq M$. Without surprise, the order on $M$ or $L$ for SPA is higher than that for SPICE. But the positive side is that the complexity of SPA does not depend on the grid size $\widetilde{N}$ which is typically much greater than $M$ in array processing. As a result, SPA can be possibly faster than SPICE if a dense sampling grid (large $\widetilde{N}$) is adopted in SPICE for obtaining high accuracy, and vise versa. Note that the order on $M$ or $L$ might be decreased in the future for both SPA and SPICE if there are more sophisticated algorithms but the linear dependence on $\widetilde{N}$ for SPICE probably cannot due to the discretization (similar results hold for other discretization-based methods).

\subsection{Proposed SPICE-PP}
Like SPA, we modify SPICE and obtain the parameter estimate within three steps. Firstly, $\m{R}$ is estimated as
\equ{\widehat{\m{R}}=\m{A}\sbra{\widetilde{\m{\theta}}}\diag\sbra{\widetilde{\m{p}}^*} \m{A}^H\sbra{\widetilde{\m{\theta}}}+\diag\sbra{\m{\sigma}^*}}
given the solution $\sbra{\widetilde{\m{p}}^*, \m{\sigma}^*}$ of the original SPICE. Then the postprocessing is applied to decompose $\widehat{\m{R}}$ into $\widehat{\m{R}}=T\sbra{\widehat{\m{u}}}+\diag\sbra{\widehat{\m{\sigma}}}$ by exploiting $K\leq M-1$. Finally, the parameter estimate is obtained via the Vandermonde decomposition of $T\sbra{\widehat{\m{u}}}$. The modified SPICE algorithm is named as SPICE with the postprocessing, abbreviated as SPICE-PP. Note that the frequency estimate of SPICE-PP is no long constrained on the grid. It is also noted that this postprocessing technique can be possibly applied to other covariance-based methods.

\subsection{Connection to Existing Discretization-Free Methods in the Case of $N=1$}
In the limiting single-snapshot case, i.e., when $N=1$, the array processing problem is mathematically equivalent to spectral analysis \cite{stoica2005spectral}. For the latter topic, discretization-free techniques have been recently proposed in \cite{candes2013towards,bhaskar2013atomic,tang2012compressed} based on atomic norm minimization (see \cite{chandrasekaran2012convex}) which, also called total variation norm in \cite{candes2013towards}, is a continuous version of the $\ell_1$ norm.
It is noted that the SPA method proposed in this paper can also be applied to the single-snapshot case. In fact, the SDPs in (\ref{formu:SDP_Nsmall}) and (\ref{formu:SDP_Nsmall_identical}) can be further simplified in such case. Given (\ref{formu:SDP_Nsmall_identical}) as an example. It follows from $\widetilde{\m{R}}=\m{y}\m{y}^H$ that $\tr\sbra{\widetilde{\m{R}}\m{R}^{-1}\widetilde{\m{R}}}=\twon{\m{y}}^2 \m{y}^H\m{R}^{-1}\m{y}$. So, an alternative formulation of (\ref{formu:SDP_Nsmall_identical}) will be
\equ{\begin{split}
&\min_{x,\m{u}} x+ \tr\sbra{T\sbra{\m{u}}}, \\
&\st \begin{bmatrix}x& \twon{\m{y}}\m{y}^H \\ \twon{\m{y}}\m{y} & T\sbra{\m{u}} \end{bmatrix}\geq\m{0}.\end{split} \label{formu:SDP_N1}}
It is interesting to note that, though obtained from a very different technique, (\ref{formu:SDP_N1}) is quite similar to the SDP formulations in \cite{candes2013towards,bhaskar2013atomic,tang2012compressed}. A detailed investigation of the relation is beyond the scope of this paper and will be posed as a future work since this paper is focused on the array processing applications in which the multisnapshot case is of the main interest.

\section{Numerical Simulations} \label{sec:simulation}
\subsection{Simulation Setups}
In this section we illustrate the performance of the proposed SPA method and compare it with existing methods via numerical simulations. The methods that we consider include SPICE \cite{stoica2011spice}, SPICE-PP, IAA \cite{yardibi2010source}, MUSIC and OGSBI-SVD \cite{yang2013off}. SPICE is a semiparametric method which can be roughly considered as a discretized version of SPA as described in Subsection \ref{sec:connectionSPICE}. SPICE-PP is a modified version of SPICE by incorporating the postprocessing technique presented in this paper. IAA is an enhanced nonparametric method. MUSIC is a classical subspace-based parametric method. OGSBI-SVD is a semiparametric method for off-grid DOA estimation. The information of source number, $K$, is required in MUSIC but not in the other methods. Since SPA and SPICE operate in different manners in the cases of equal/different noise variances, `+' is used to indicate the case when the equal noise variances assumption is imposed. For example, SPA+ refers to SPA with the assumption. Without ambiguity, ``SPA'' can refer to either the collectively called SPA technique or the SPA method with different noise variances (in contrast to SPA+) hereafter (similarly for the use of ``SPICE''). Some setups of the algorithms above are as follows. The SDPs of SPA are implemented using CVX with the SDPT3 solver \cite{grant2008cvx,toh1999sdpt3}. SPICE is implemented as in \cite{stoica2011spice} and terminated when the relative change of the objective function value in two consecutive iterations falls below $1\times 10^{-6}$ or the maximum number of iterations, set to 500, is reached. IAA is terminated if the relative change of the $\ell_2$ norm of the power vector in two consecutive iterations falls below $1\times 10^{-6}$ or the maximum number of iterations, set to 500, is reached. OGSBI-SVD is implemented as in \cite{yang2013off} except that the source number is unknown (see details in \cite{yang2013analysis}).

\subsection{Spectra Comparisons} \label{sec:spectra}

\begin{figure*}
\centering
  \subfigure[$\sbra{N, \text{SNR}}=\sbra{200,20\text{dB}}$]{
    \label{Fig:spectra1}
    \includegraphics[width=3.15in]{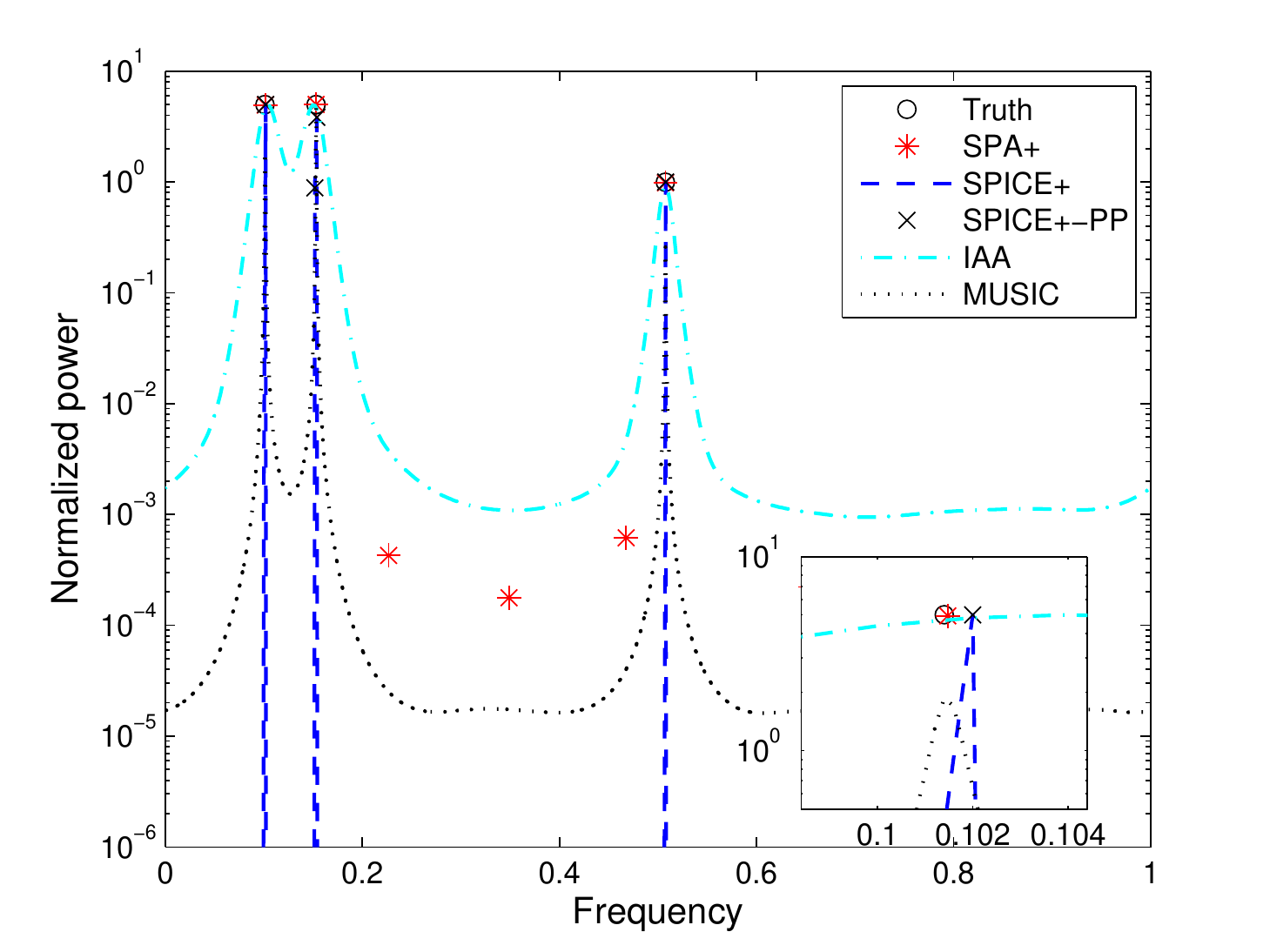}} %
  \subfigure[$\sbra{N, \text{SNR}}=\sbra{200,0\text{dB}}$]{
    \label{Fig:spectra2}
    \includegraphics[width=3.15in]{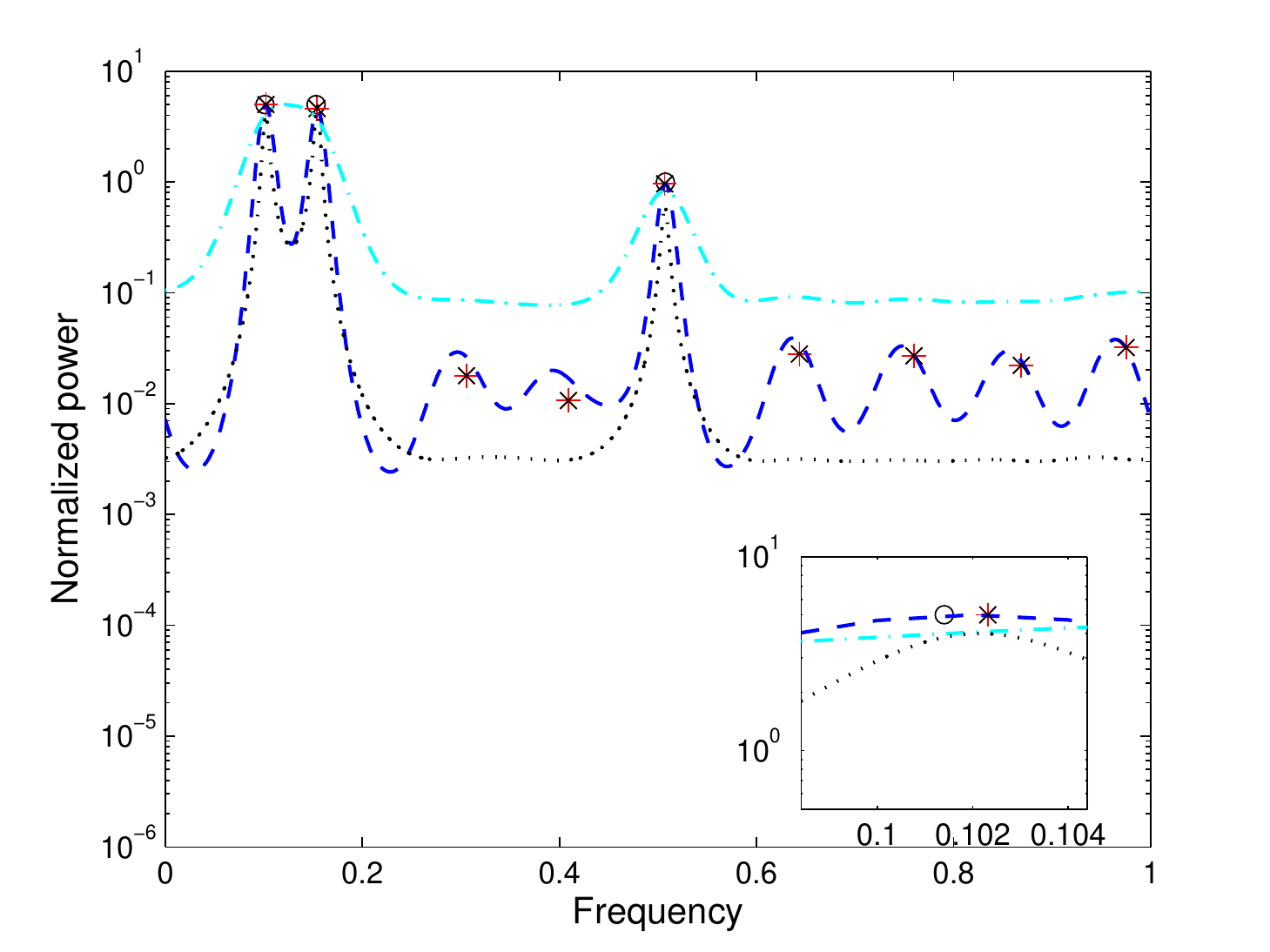}}
  \subfigure[$\sbra{N, \text{SNR}}=\sbra{+\infty,-20\text{dB}}$]{
    \label{Fig:spectra3}
    \includegraphics[width=3.15in]{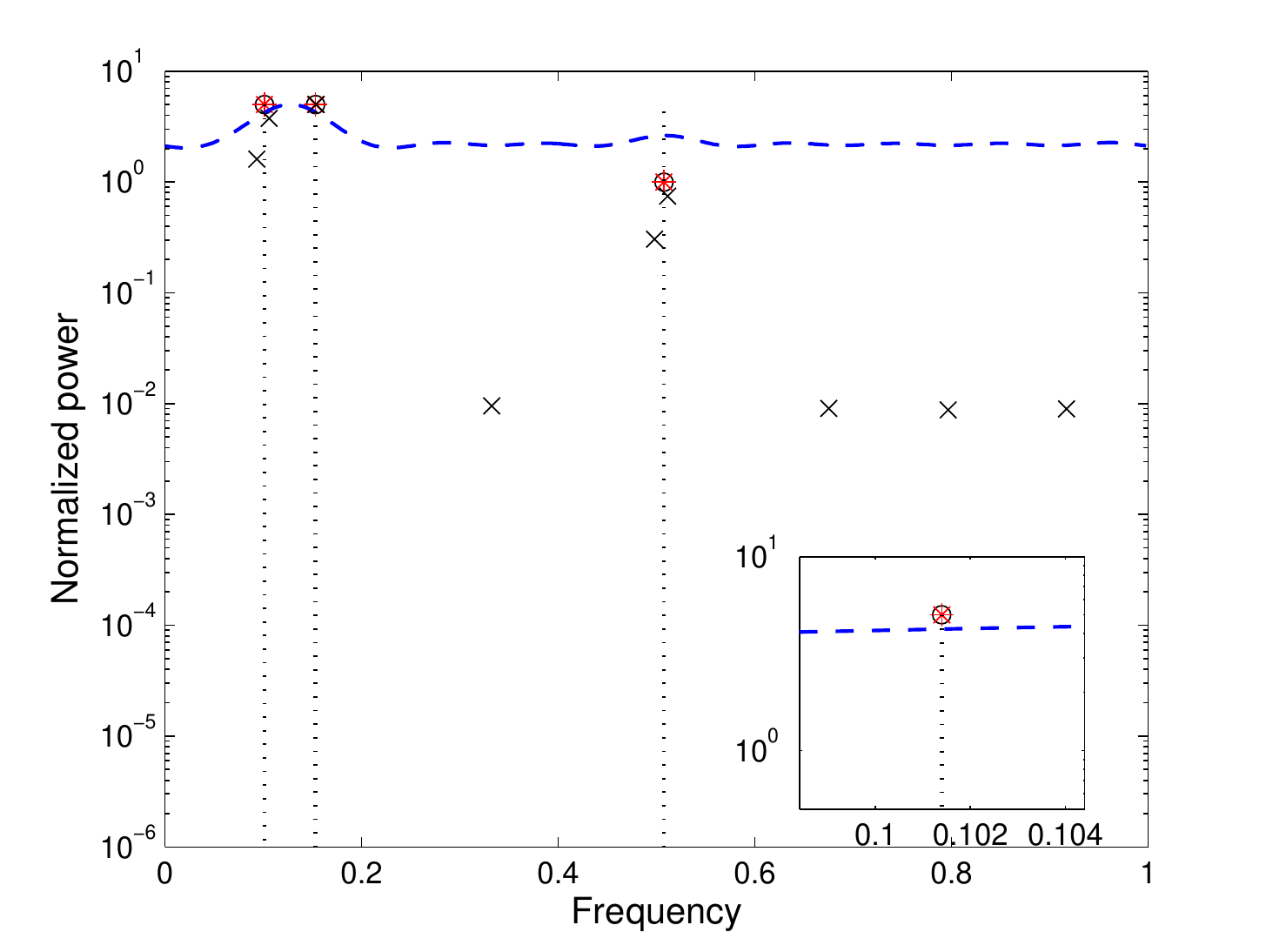}} %
  \subfigure[$\sbra{N, \text{SNR}}=\sbra{1,60\text{dB}}$]{
    \label{Fig:spectra4}
    \includegraphics[width=3.15in]{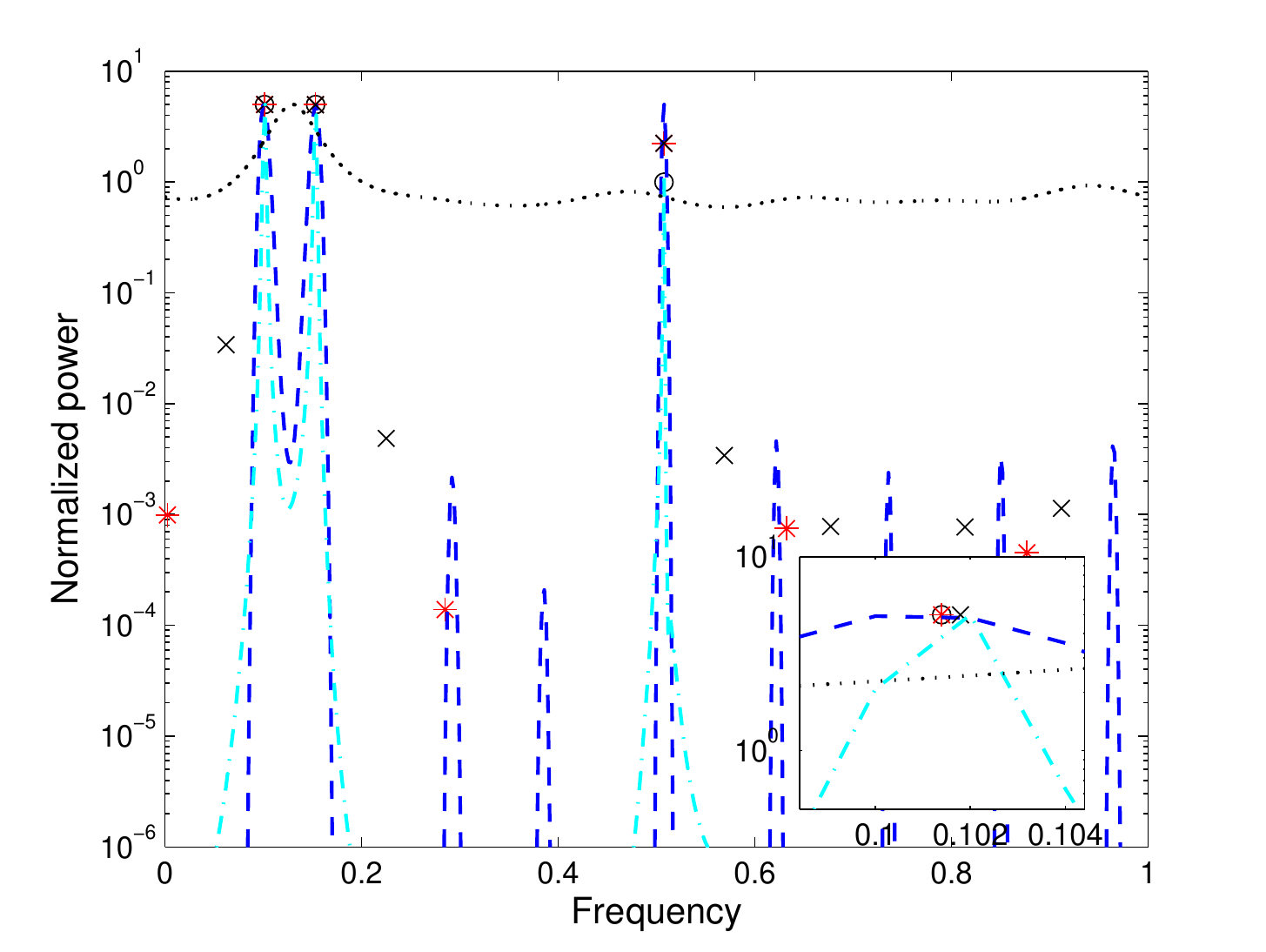}}
\centering
\caption{Spectra of SPA+, SPICE+, SPICE+-PP, IAA and MUSIC for uncorrelated sources with respect to $\sbra{N, \text{SNR}}$. Other settings include: $\m{\theta}^o=\mbra{0.1014, 0.1532, 0.5077}^T$, $\m{p}^o=\mbra{5, 5, 1}^T$ and ULA with $M=10$. The subfigure in the lower right corner zooms in the area around source 1.} \label{Fig:spectra}
\end{figure*}

\begin{figure*}
\centering
  \subfigure[$\sbra{N, \text{SNR}}=\sbra{200,20\text{dB}}$]{
    \label{Fig:spectra_coh1}
    \includegraphics[width=3.15in]{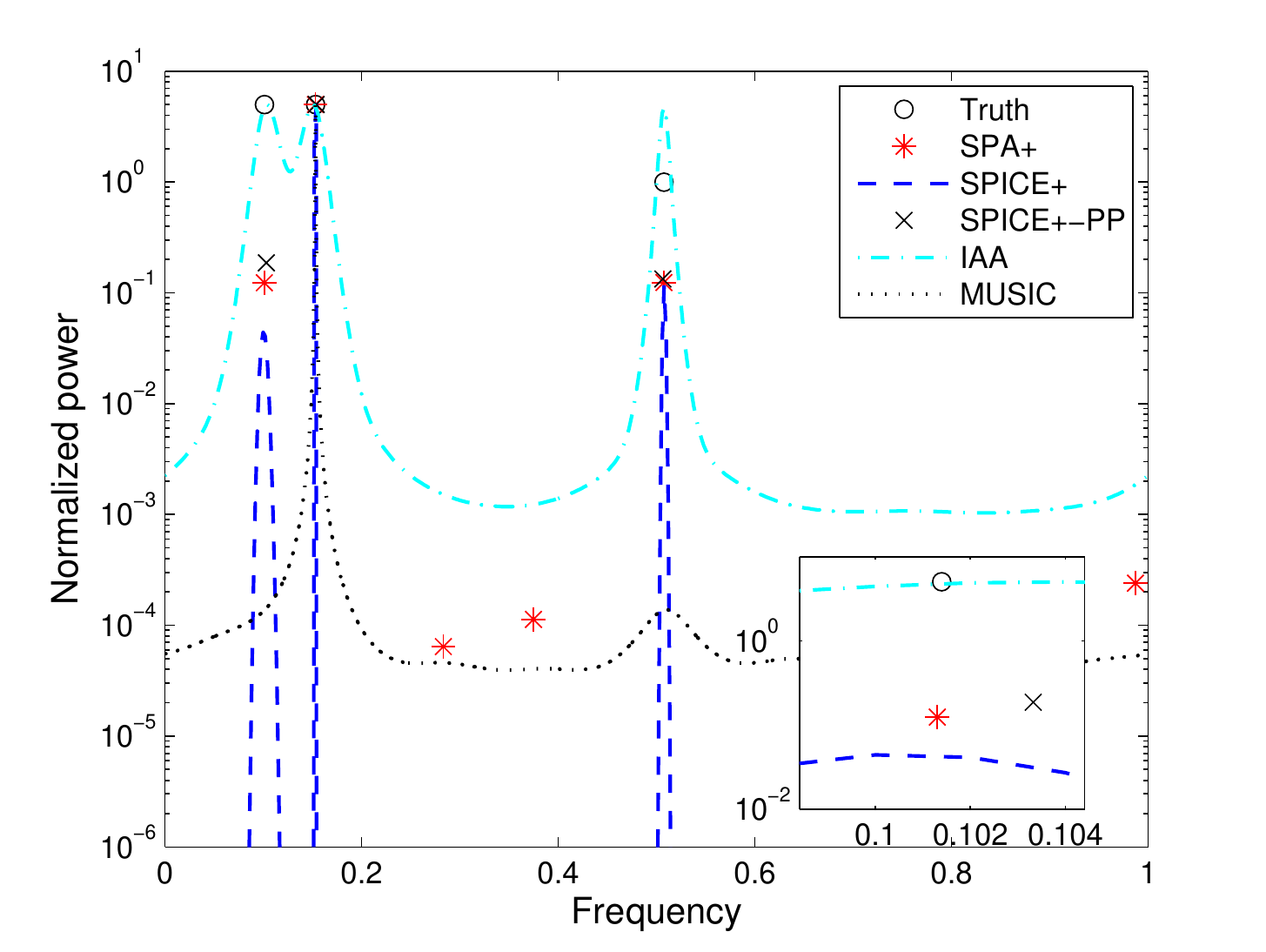}} %
  \subfigure[$\sbra{N, \text{SNR}}=\sbra{200,0\text{dB}}$]{
    \label{Fig:spectra_coh2}
    \includegraphics[width=3.15in]{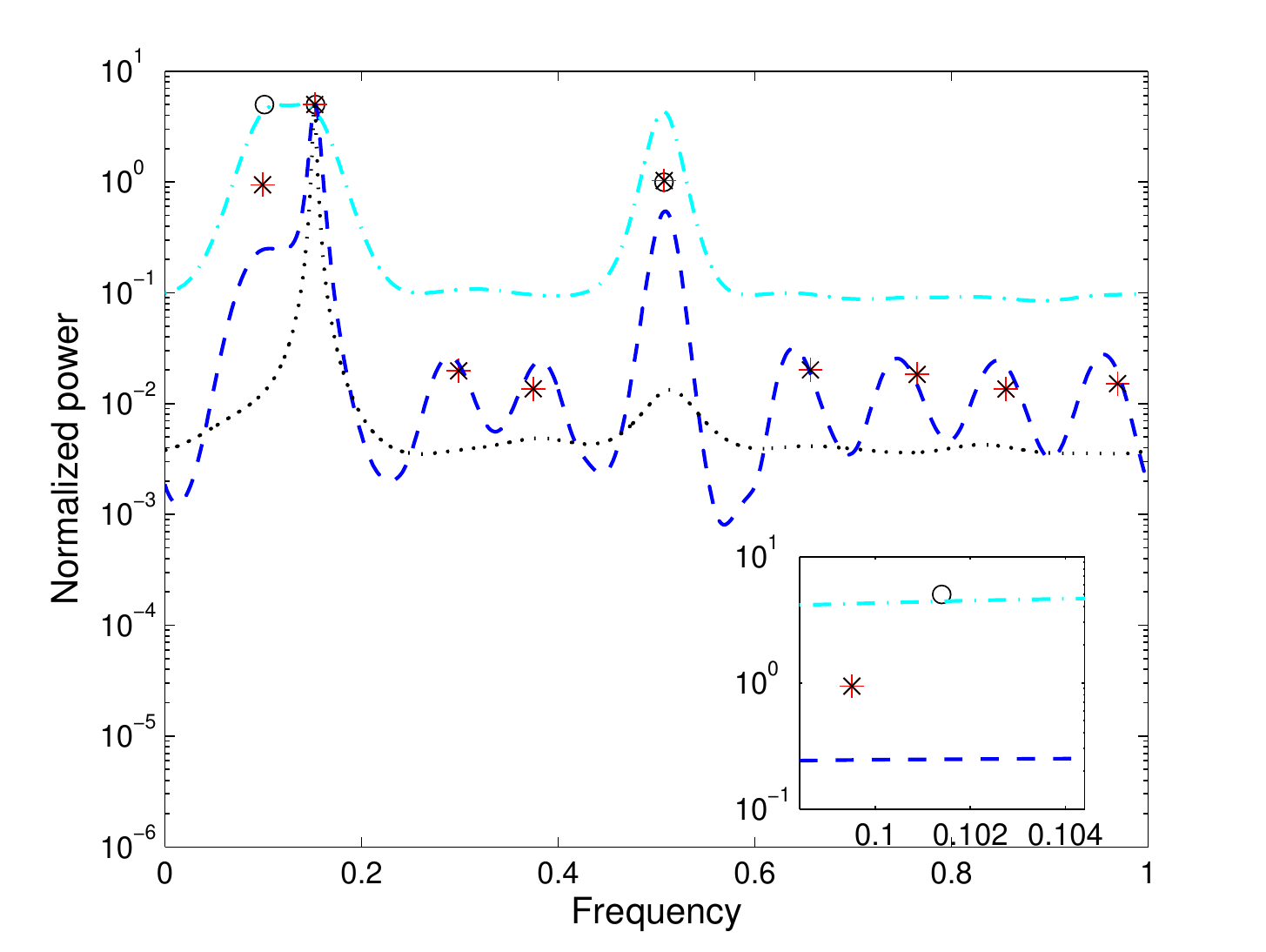}}
\centering
\caption{Spectra of SPA+, SPICE+, SPICE+-PP, IAA and MUSIC for coherent sources with respect to $\sbra{N, \text{SNR}}$, where source 3 is a replica of source 1. Other settings include: $\m{\theta}^o=\mbra{0.1014, 0.1532, 0.5077}^T$, $\m{p}^o=\mbra{5, 5, 1}^T$ and ULA with $M=10$. The subfigure in the lower right corner zooms in the area around source 1.} \label{Fig:spectra_coh}
\end{figure*}

We compare spectra of the aforementioned methods in this subsection. In our simulation, we consider $K=3$ uncorrelated/coherent sources with power $\m{p}^o=\mbra{5, 5, 1}^T$ from directions specified by the frequency vector $\m{\theta}^o=\mbra{0.1014, 0.1532, 0.5077}^T$. A ULA with $M=10$ is used to receive the signals. Each of the source signals is randomly generated with constant amplitude and random phase, which is usually the situation in communications applications \cite{stoica2011spice}. Zero-mean white circular Gaussian noise is added with the noise variance $\sigma^o$. The signal to noise ratio (SNR) is defined as the ratio of the minimum source power to the noise variance (in dB), i.e., $\text{SNR}=10\log_{10}\frac{\min\sbra{p_j^o}}{\sigma^o}$. The grid number is set to 500 for SPICE and IAA.

Our simulation results of uncorrelated sources are presented in Fig. \ref{Fig:spectra} with respect to different settings of $\sbra{N, \text{SNR}}$, where some curves are omitted for better visual effects. In the case of large snapshots, e.g., $N=200$ as shown in Figs. \ref{Fig:spectra1} and \ref{Fig:spectra2}, we observe the following phenomena. IAA produces a dense spectrum and exhibits significant resolution degradation in the moderate/low SNR regime. It cannot separate the first two sources in Fig. \ref{Fig:spectra2}. SPICE+ produces a sparse spectrum which contains only a few spikes in the high SNR regime, e.g., $\text{SNR}=20$dB. However, it produces a dense spectrum due to the identifiability problem as addressed in Subsection \ref{sec:connectionSPICE} in the moderate/low SNR regime, e.g., $\text{SNR}=0$dB, where it returns a zero estimate of the noise variance. Unlike SPICE+, SPICE always produces a dense spectrum whenever the SNR is (its spectra are omitted in Fig. \ref{Fig:spectra}). Without surprise, the SPA method proposed in this paper always produces a sparse spectrum. Small spurious spikes exhibit due to the presence of noise and absence of the knowledge of source number and the noise level. By using the postprocessing technique presented in this paper both SPICE+-PP and SPICE-PP produce sparse spectra, and their differences from SPA+ (or SPA) are caused by the modeling error of SPICE as described in Subsection \ref{sec:connectionSPICE}. Moreover, SPICE+ and SPICE+-PP have the same spectrum in the high SNR regime since the postprocessing does not alter the spectrum when SPICE+ has already produced a sparse spectrum. This point will be revisited later.

The limiting case $N=+\infty$ with a very low SNR ($-20$dB) is studied in Fig. \ref{Fig:spectra3}, where the true covariance matrix is adopted to implement SPA, SPICE and MUSIC. IAA is not considered in this scenario. It is shown that SPA (similarly for MUSIC) can exactly localize the three sources, which verifies the conclusion of Theorem \ref{thm:consistency} that SPA is statistically consistent. SPICE+ has low resolution in the low SNR regime due to the identifiability problem and cannot separate the first two sources. Due to the modeling error of SPICE, SPICE+-PP is not consistent as well. Though this paper is mainly focused on the case of moderate/large snapshots, it is noted that SPA can be applied to the case of small or even a single snapshot as shown in Fig. \ref{Fig:spectra4} where MUSIC fails.

Fig. \ref{Fig:spectra_coh} presents simulation results of coherent sources, where source 3 is exactly a replica of source 1. It is shown that the proposed SPA method has consistently good performance in the presence of complete correlation due to the adopted covariance fitting criterion as shown in \cite{stoica2011new,stoica2011spice}. In contrast, IAA has low resolution as in the case of uncorrelated sources. MUSIC typically misses coherent sources. SPICE tends to miss coherent sources as well in the moderate/low SNR regime. Finally, note that the power estimates of the coherent sources are attenuated in SPA, the reason of which should be investigated in the future.

\subsection{Quantitative Comparisons with SPICE}
Quantitative comparisons will be carried out in this subsection to demonstrate advantages of our discretization-free SPA method compared to SPICE. Readers are referred to \cite{stoica2011spice} for performance comparisons of SPICE with IAA and MUSIC. Unlike SPA, SPICE might produce a dense spectrum as the nonparametric methods. The frequency estimate of SPICE is obtained using the peaks of the spectrum following from \cite{stoica2011spice}. To illustrate effects of the discretization adopted in SPICE, three discretization levels are considered with the grid size $\widetilde{N}=200, 500, 1000$, respectively. For convenience, the SPICE algorithm adopting the three discretization schemes will be referred to as SPICE1, SPICE2 and SPICE3, respectively. The `+' symbol will be used as before. Metrics recorded include mean squared error (MSE) and CPU time usage. The MSE of the frequency estimation is computed as $\frac{1}{K}\twon{\widehat{\m{\theta}}_K-\m{\theta}^o}^2$ and then averaged over a number of Monte Carlo runs, where $\widehat{\m{\theta}}_K$ denotes the frequency estimate which is obtained by keeping the associated largest $K$ entries of the power estimate $\widehat{\m{p}}$. The CRLB is commonly used as a benchmark when evaluating the performance of various estimators though it is a lower bound for only unbiased estimators. Note also that to compute the CRLB requires the knowledge of $K$ which is not used in SPA and SPICE. So, there might exist a gap between the CRLB and the performance of SPA or SPICE that we study. The simulations study both the ULA and SLA cases and are focused on uncorrelated sources.

\subsubsection{The ULA Case} {\em Experiment 1} studies performance variation with respect to the SNR. We consider a ULA with $M=10$. Without loss of generality, $K=2$ uncorrelated sources impinge on the array with (off-grid) frequencies $\frac{1}{6}$ and $\frac{4}{15}$ and unit powers. Notice that each frequency is a third grid interval away from the nearest grid point for SPICE. Since the best frequency estimate for a given source is the nearest grid point, the MSE of the frequency estimation of SPICE is lower bounded by $\frac{1}{9\widetilde{N}^2}$ regardless of the SNR.\footnote{Under the assumption that the source is randomly located in one or more grid intervals, the lower bound will be $\frac{1}{12\widetilde{N}^2}$ \cite{yang2013off}.} The number of snapshots is set to $N=200$ and the SNR varies in $\lbra{-20, -15, \dots, 25}$dB. 200 Monte Carlo runs are used for each algorithm to obtain the metrics, where the source signals and the noise are both i.i.d.~Gaussian. Fig. \ref{Fig:MSE_ULA} plots MSEs of the simulation results, where the two cases of equal and different noise variances are separately presented to provide a better illustration. SPICE has a better performance with a finer discretization but lower bounded by some constant as mentioned above. To the contrary, the MSEs of the SPA methods improve constantly with the SNR and gradually approach the CRLBs. Both the SPA methods and the SPICE methods have similar performance trends in the two cases of equal and different noise variances. However, SPICE-PPs perform differently. In the former case, SPICE+-PPs and SPICE+s coincide in the high SNR regime since they produce the same sparse spectra as shown in Subsection \ref{sec:spectra}. In the latter case, SPICE-PPs outperform the associated SPICEs when the SNR is larger than some threshold thanks to the postprocessing technique presented in this paper. But when the SNR is sufficiently high, the modeling error caused by the discretization dominates the total uncertainties and further performance improvement is impossible. On the other hand, in the moderate/low SNR regime where the measurement noise dominates the uncertainties, SPICE-PP coincides with SPA as expected. Note that the bad performances of SPICE2+-PP and SPICE3+-PP at $\text{SNR}=5$dB are caused by very few outliers (1 and 4 trials, respectively, out of 200), where SPICE+ might not converge within 500 iterations and result in a less accurate $\widehat{\m{R}}$ that is used for parameter estimation. Finally, notice that SPICE can possibly outperform SPA in the SNR range $\mbra{-10, 10}$dB (depending on the discretization level), which will be discussed in Subsection \ref{sec:discussion}.

\begin{figure}
\centering
  \includegraphics[width=3.5in]{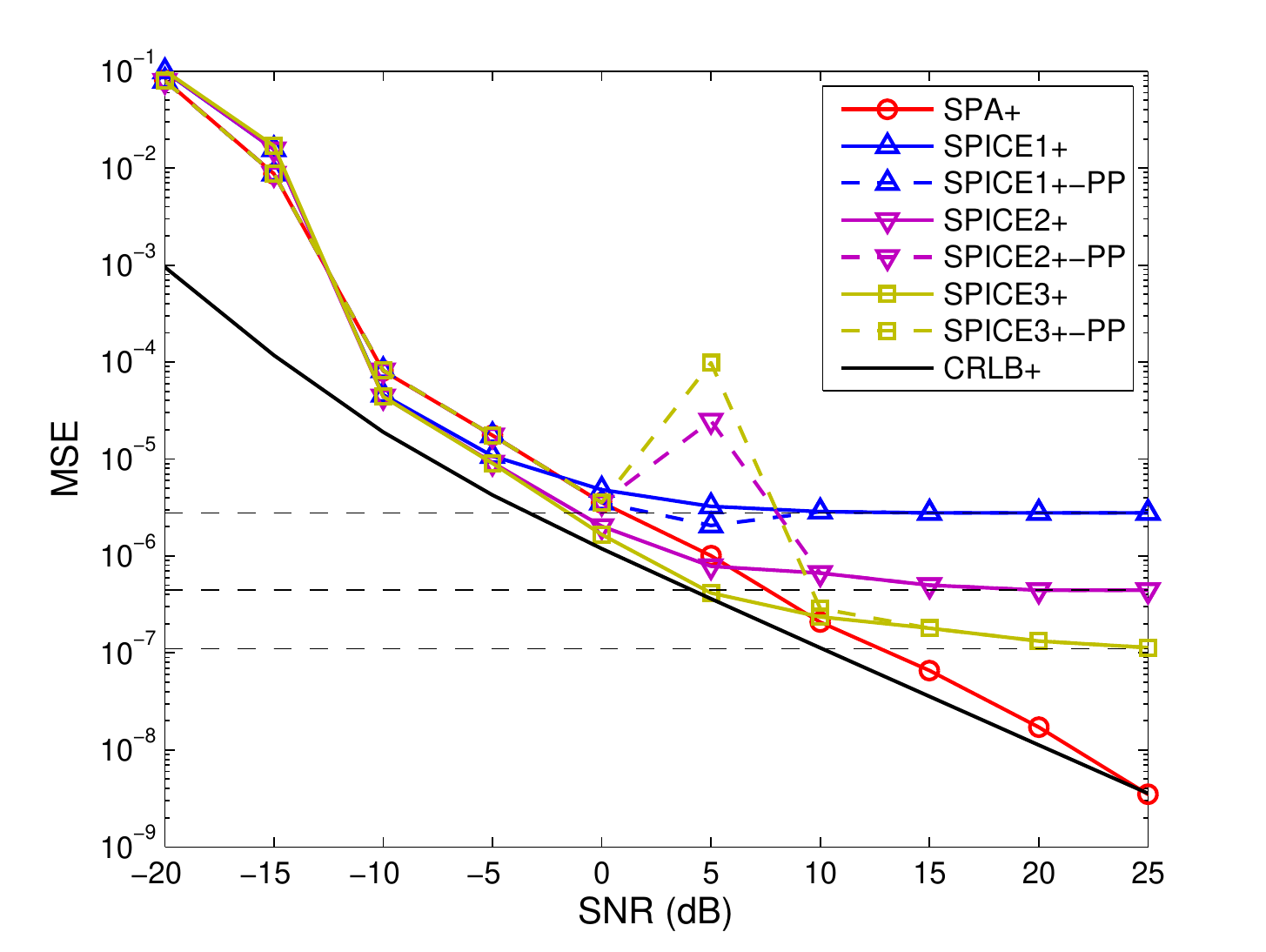}
  \includegraphics[width=3.5in]{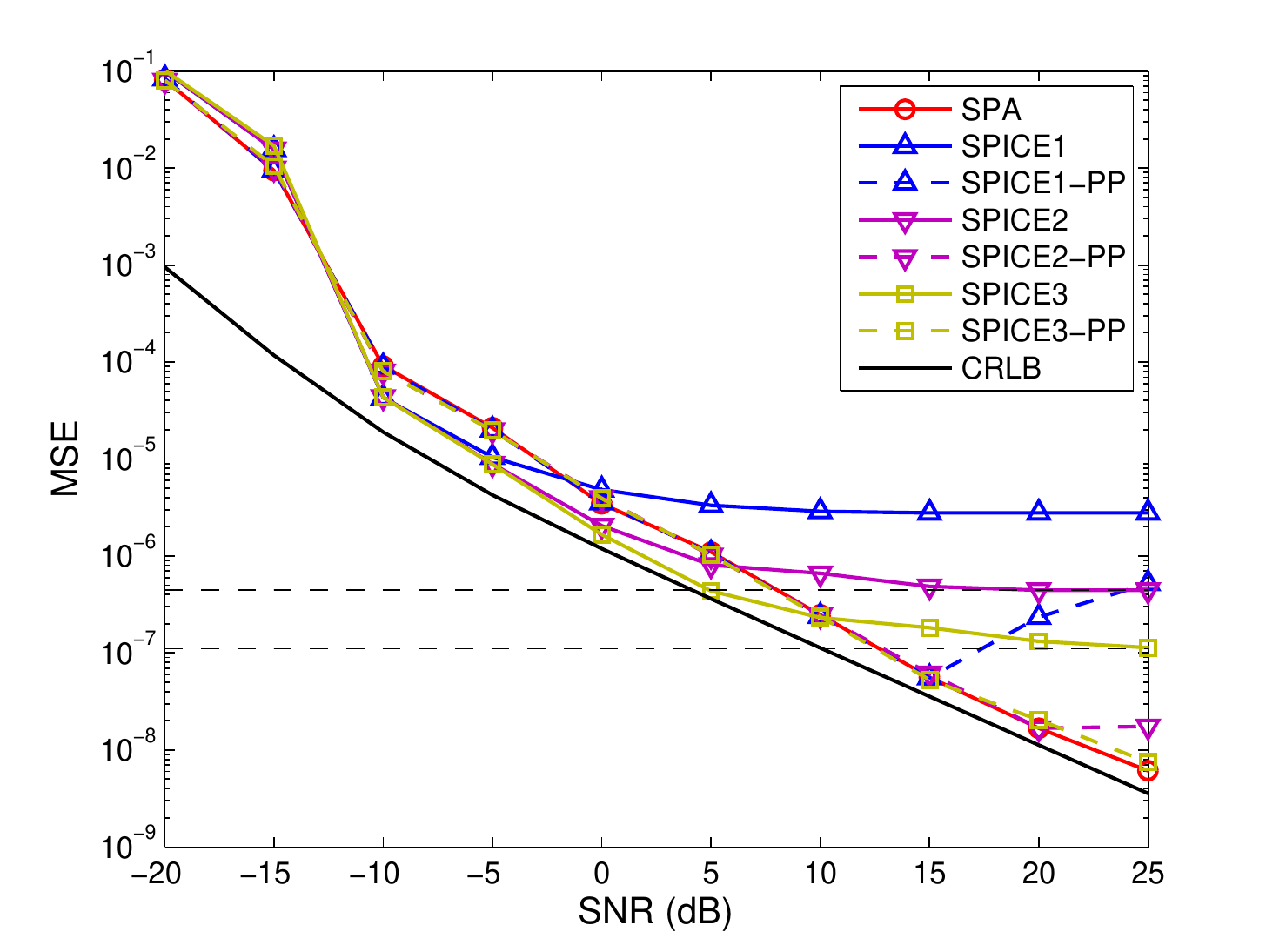}
\centering
\caption{MSEs of frequency estimates of SPA compared with SPICE and the CRLB. Some settings include: ULA with $M=10$, $K=2$ uncorrelated sources with $\m{\theta}^o=\mbra{\frac{1}{6}, \frac{4}{15}}^T$ and $\m{p}^o=\mbra{1,1}^T$, and $N=200$. The horizontal dashed lines are lower bounds of the SPICEs due to the discretization.} \label{Fig:MSE_ULA}
\end{figure}

Fig. \ref{Fig:Time_ULA} presents CPU times of {\em Experiment 1}. Since the postprocessing can be applied efficiently, the time usage of SPICE-PP is slightly longer than SPICE and is omitted. Both SPICE+s and SPICEs have similar performance trends with different discretization levels because the number of iterations used is approximately the same. As a result, the time usage of SPICE is proportional to the grid size $\widetilde{N}$ while the proposed discretization-free SPA does not depend on $\widetilde{N}$. Fig. \ref{Fig:Time_ULA} shows that when $\widetilde{N}=1000$ SPICE+ and SPICE are constantly slower than SPA+ and SPA, respectively. When $\widetilde{N}=500$, SPICE+ is also slower than SPA+ sometimes.

\begin{figure}
\centering
  \includegraphics[width=3.5in]{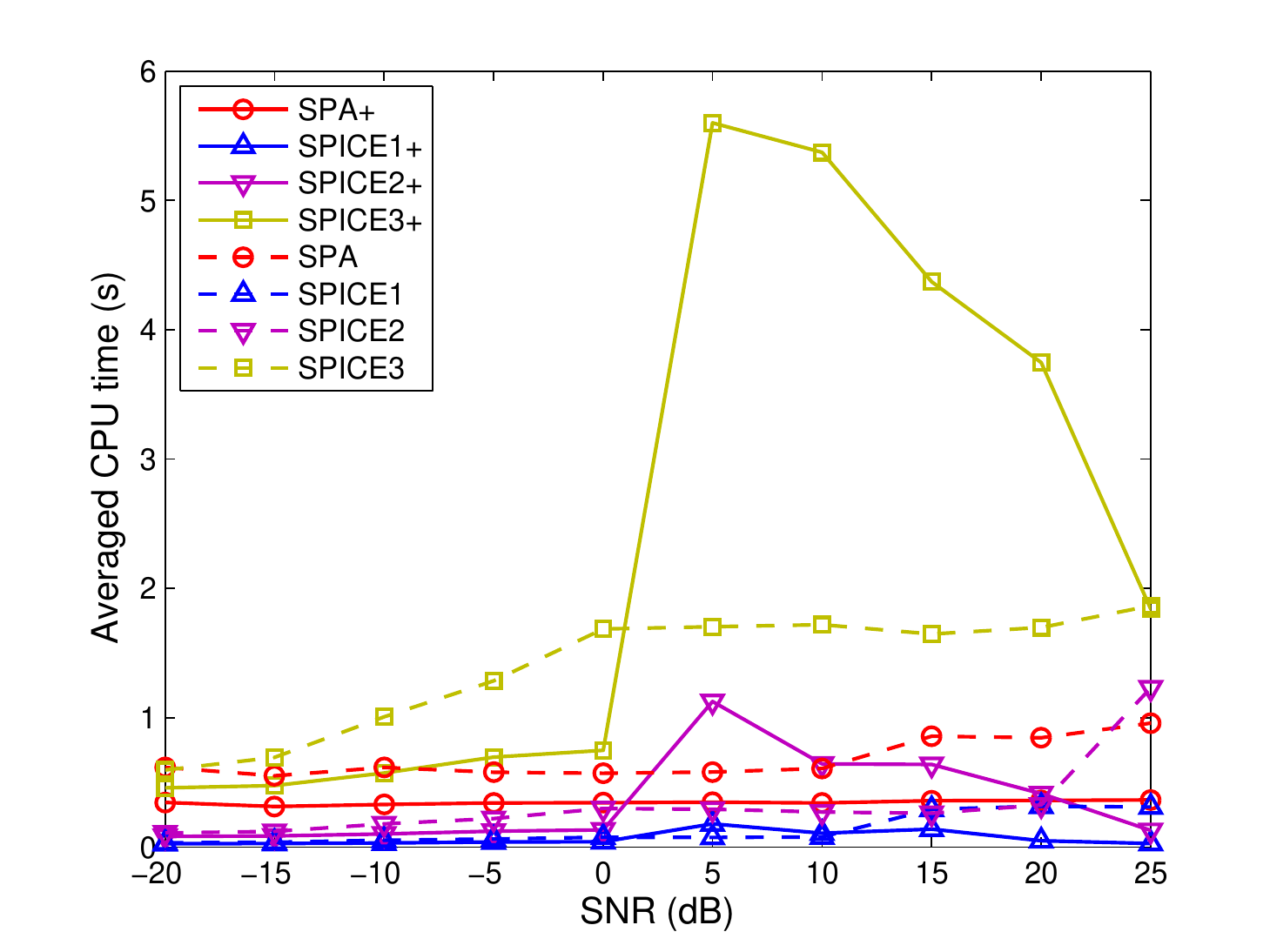}
\centering
\caption{CPU time usages of SPA compared with SPICE for estimating $K=2$ uncorrelated sources. Some settings include: ULA with $M=10$, $\m{\theta}^o=\mbra{\frac{1}{6}, \frac{4}{15}}^T$, $\m{p}^o=\mbra{1,1}^T$ and $N=200$.} \label{Fig:Time_ULA}
\end{figure}

%

{\em Experiment 2} studies the performance with respect to the array length $M$. We repeat {\em Experiment 1} but set $\text{SNR}=10$dB and vary $M$ in $\lbra{5, 10, \dots, 40}$. Moreover, we consider only the case $\widetilde{N}=1000$ for SPICE. Only the case where $M\leq20$ is considered for SPA (excluding SPA+) due to time consideration. The simulation results are presented in Fig. \ref{Fig:ULA_varyM}. Since CRLB+ and CRLB are slightly different and almost undistinguishable, only CRLB+ is plotted. As before, SPICE3+, SPICE3 and SPICE3+-PP share the same lower bound due to discretization and the property of the SPICE+ estimator, while SPA+, SPA and SPICE3-PP can outperform this bound. As $M$ increases, e.g., when $M\geq25$, the performances of SPA and SPICE3-PP hardly improve and the gaps between the algorithms and CRLB+ become larger. Two possible reasons are as follows: 1) Unlike CRLB+, SPA does not use the knowledge of $K$ but $K\leq M-1$, which becomes rougher when $M$ increases and $K$ keeps unaltered, and 2) The modeling error of SPICE increases as $M$ increases. The figure at the bottom indicates that the speed of SPA+ scales well with $M$ in our considered scenario and is faster than SPICE3+ as $M\leq 30$ and constantly faster than SPICE3. The number of iterations of SPICE3+ is empirically observed to decrease with increasing array length in our considered scenario, leading to a seeming strange result that SPICE3+ gets faster when $M$ increases.

\begin{figure}
\centering
  \includegraphics[width=3.5in]{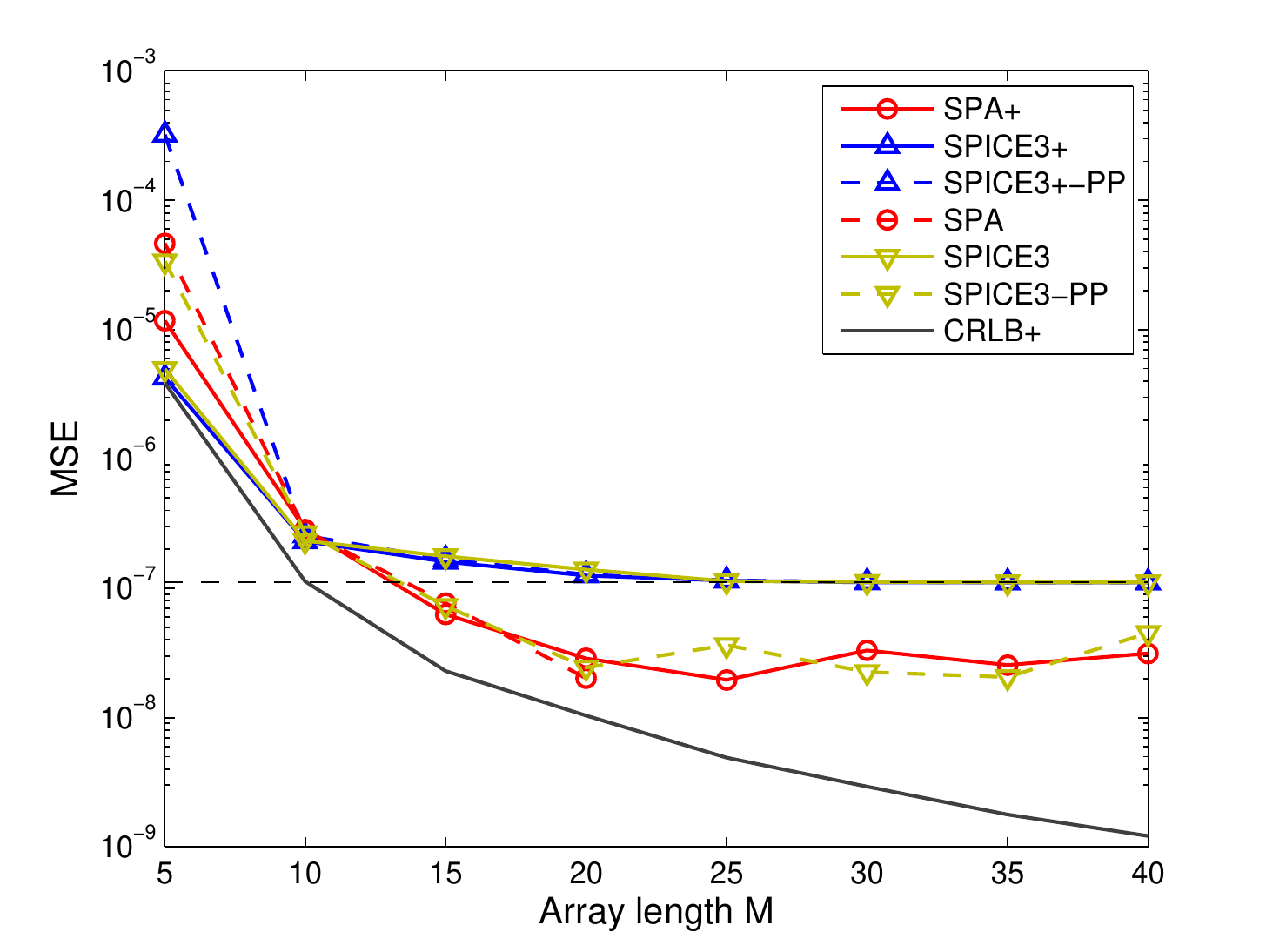}
  \includegraphics[width=3.5in]{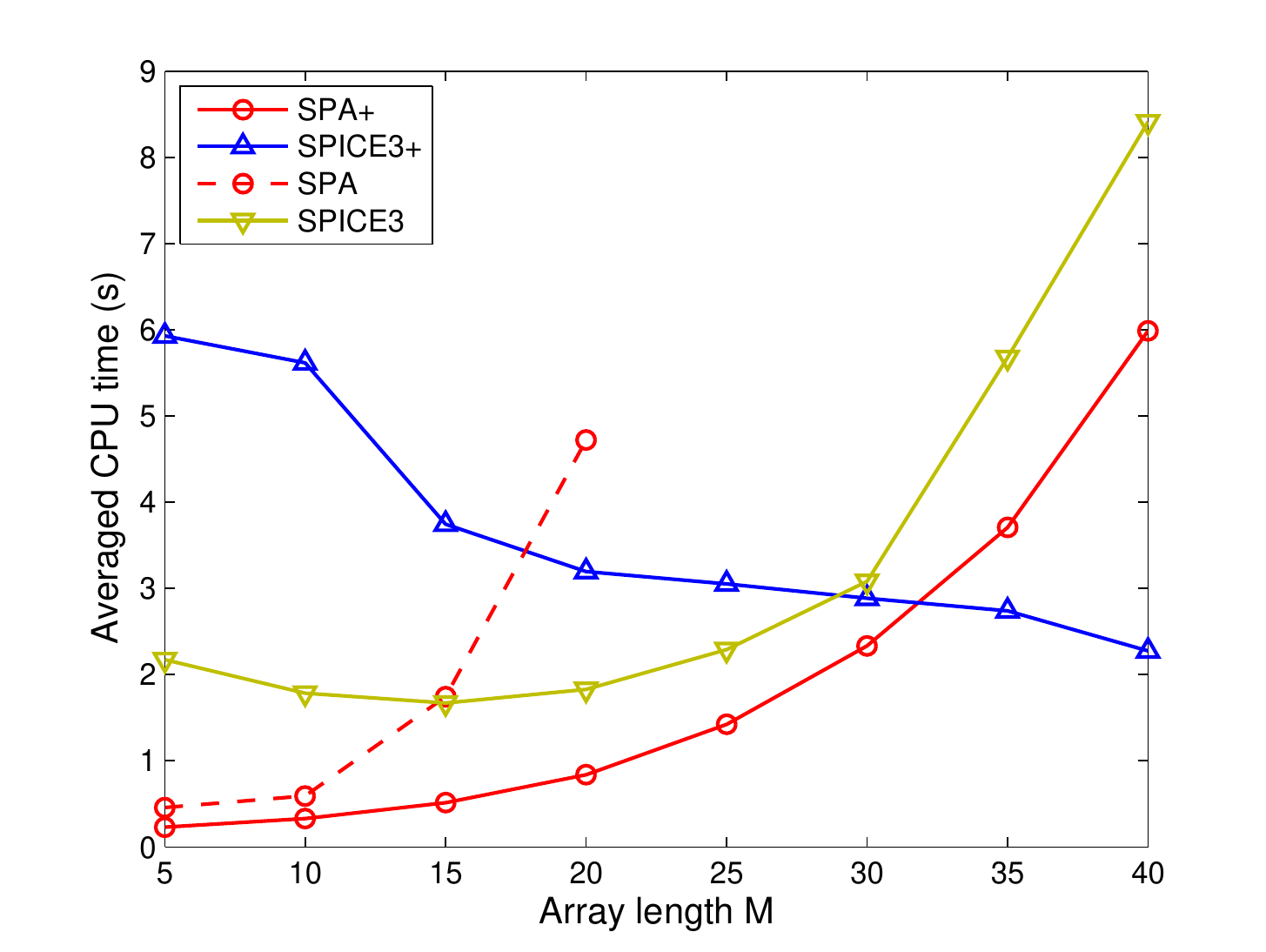}
\centering
\caption{MSEs of frequency estimates (top) and CPU time usages (bottom) of SPA with an $M$-element ULA compared with SPICE. Some settings include: $K=2$ uncorrelated sources with $\m{\theta}^o=\mbra{\frac{1}{6}, \frac{11}{30}}^T$ and $\m{p}^o=\mbra{1,1}^T$, $\text{SNR}=10\text{dB}$ and $N=200$.} \label{Fig:ULA_varyM}
\end{figure}

\subsubsection{The SLA Case} {\em Experiment 3} studies the SLA case where a 4-element redundant array $\m{\Omega}=\lbra{1,2,5,7}$ is considered. We try to verify Theorem \ref{thm:efficiency} and attempt to locate maximally $K=6$ uncorrelated sources with the frequency vector $\m{\theta}^o=\mbra{0.1008,0.1809,0.4001, 0.5509, 0.7006, 0.8501}^T$ and power vector $\m{p}^o=\mbra{2, 2, 2, 1, 1, 1}^T$. Moreover, we set $\text{SNR}=10$dB and vary the number of snapshots $N$ in $\lbra{20,200,2000,20000}$. 1000 Monte Carlo runs are used to obtain the metrics for each $N$. To evaluate the performance, we calculate the ratio $\frac{\text{CRLB+}}{\text{MSE}}$ for each algorithm at each $N$. For an unbiased estimator, the ratio is called its efficiency. The larger the ratio is ($\leq1$ for an unbiased estimator), the more accurate the estimator will be. Table. \ref{table:SLA_frequency} presents the simulation results. Remarkably, SPA+ (or SPA) can outperform the CRLB+ (or CRLB) when $N\geq200$ since the ratios are larger than 1 (or 0.9281). Though SPICE+ (or SPICE) can outperform the CRLB+ (or CRLB) as well at $N=200$, its gap to the CRLB+ (or CRLB) becomes larger as $N$ increases. Moreover, SPICE+-PP (or SPICE-PP) is consistently better than SPICE+ (or SPICE) and worse than SPA+ (or SPA). Notice that as $N$ gets larger, the gap between SPA+ (or SPA) and 1 (or 0.9281) becomes smaller, which is consistent with the conclusions of Theorem \ref{thm:efficiency} that $\frac{\text{CRLB+}}{\text{MSE}_{\text{SPA+}}}, \frac{\text{CRLB}}{\text{MSE}_{\text{SPA}}}\rightarrow1$, as $N\rightarrow+\infty$. It is also noted that the source power and noise variance estimates of the proposed SPA method have similar performances and their metrics are omitted. The frequency and power estimates of SPA+ at $N=200$ are plotted in Fig. \ref{Fig:SLA_N200}.

\begin{table}
 \caption{MSEs of Frequency Estimates of SPA Compared with SPICE and CRLB. Each of the presented values is the ratio $\frac{\text{CRLB+}}{\text{MSE} (\text{or CRLB})}$ (the larger the better).}
 \centering
\begin{tabular}{l|r|r|r|r}
  \hline
   & $N=20$ & $N=200$ & $N=2000$ & $N=20000$\\\hline
   CRLB+ & 1  &  1  &  1  &  1\\\hline
   \textbf{SPA+} & 0.6780  &  \textbf{1.2558}  &  \textbf{1.1327}  &  \textbf{1.0370}\\\hline
   SPICE3+ & 0.4730  &  1.0419  &  0.2286  &  0.0267\\\hline
   SPICE3+-PP & 0.6610  &  1.2528  &  0.6039  &  0.1049\\\hline\hline
   CRLB & 0.9281  &  0.9281 &   0.9281  &  0.9281\\\hline
   \textbf{SPA} & 0.5615  &  \textbf{1.2134}  &  \textbf{1.0671}    &\textbf{1.0168}\\\hline
   SPICE3 & 0.5019  &  0.9535   & 0.2250  &  0.0267\\\hline
   SPICE3-PP & 0.5555   & 1.1693  &  0.5791  &  0.1048\\\hline

\end{tabular} \label{table:SLA_frequency}
\end{table}


\begin{figure}
\centering
  \includegraphics[width=3.5in]{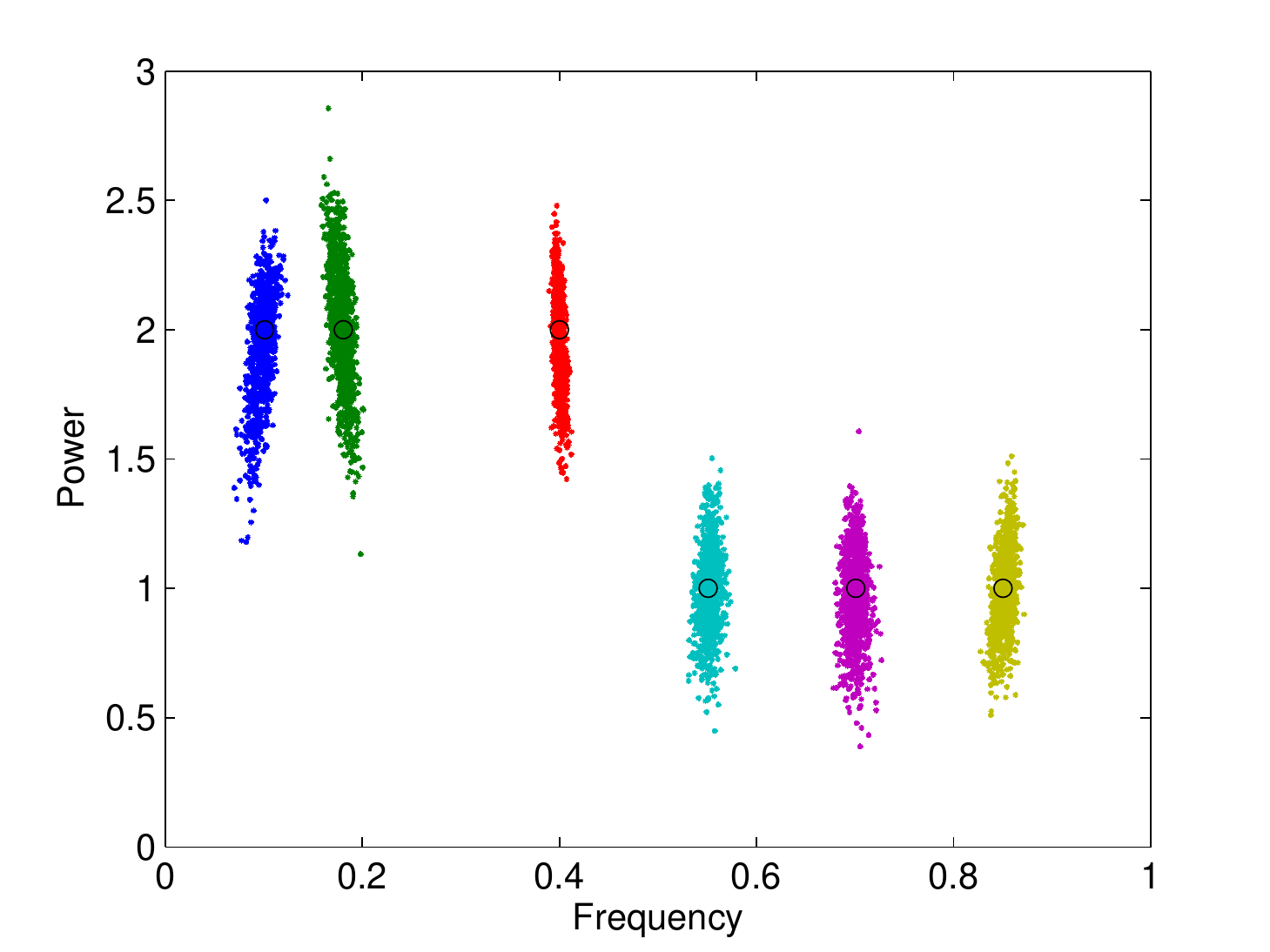}
\centering
\caption{Frequency and power estimates of SPA+ for estimating $K=6$ uncorrelated sources with a 4-element SLA $\m{\Omega}=\lbra{1,2,5,7}$ when $N=200$ and $\text{SNR}=10\text{dB}$ (1000 Monte Carlo runs). Black circles indicate the true frequencies and corresponding powers.} \label{Fig:SLA_N200}
\end{figure}

\subsection{Comparison With OGSBI-SVD}
We compare SPA with the off-grid method OGSBI-SVD in \cite{yang2013off} in this subsection. OGSBI-SVD utilizes the sparse Bayesian learning (SBL) technique \cite{tipping2001sparse,ji2008bayesian}, which mimics the ML estimation according to \cite{stoica2012spice}, and is based on an off-grid observation model which is a first-order approximation of the exact model with continuous frequencies and has a reduced modeling error (recall that the observation model of SPICE and many others is a zeroth-order approximation while SPA relies on the exact model). In OGSBI-SVD the grid offset (the distance from a true frequency location to its nearest grid point) is estimated jointly with the sparse signal and singular value decomposition (SVD) is used for reducing dimension of the observed data and faster convergence. Since the covariance $\m{R}$ is also involved in OGSBI-SVD as in SPICE which cannot be appropriately separated into the source and noise parts, OGSBI-SVD usually produces a dense spectrum like SPICE.

In \emph{Experiment 4}, we repeat {\em Experiment 1} for SPA+ and OGSBI-SVD and consider two discretization levels, $\widetilde{N}=100,200$, for OGSBI-SVD (denoted by `1', `2' respectively). The MSEs are plotted in Fig. \ref{Fig:MSE_OGSBI}. As expected, OGSBI-SVD can exceed the lower bounds of on-grid methods (horizontal dashed lines). As SNR increases, e.g., at $\text{SNR}=25$dB, the modeling error of OGSBI-SVD becomes nonnegligible. Then SPA is more accurate than OGSBI-SVD. Note also that SPA is more robust to noise. For example, SPA can produce satisfactory results at $\text{SNR}=-10$dB while OGSBI-SVD cannot. In computational speed, SPA is slightly slower than OGSBI-SVD with $\widetilde{N}=100$ and about 3 times faster with $\widetilde{N}=200$.

\begin{figure}
\centering
  \includegraphics[width=3.5in]{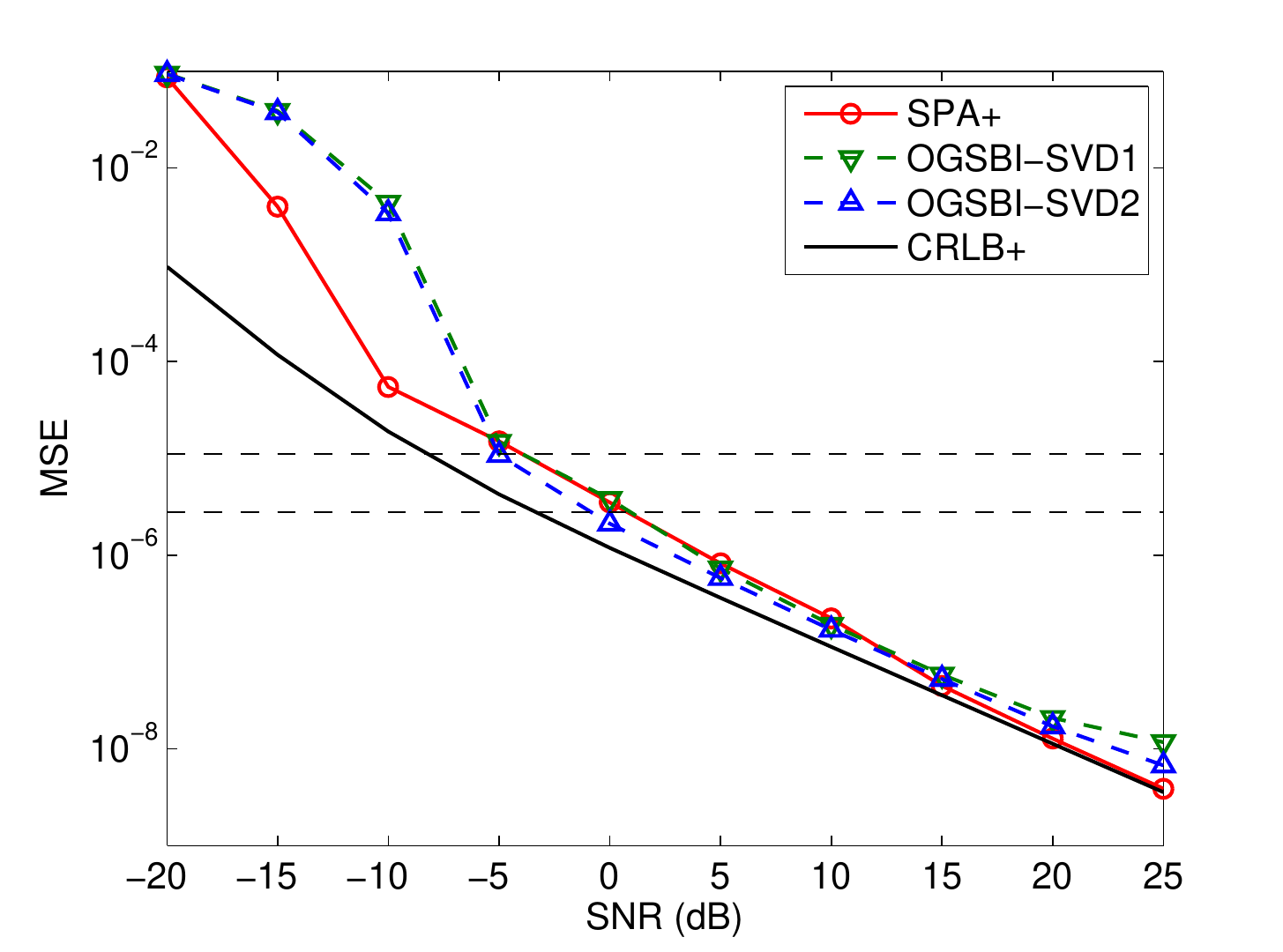}
\centering
\caption{MSEs of frequency estimates of $K=2$ uncorrelated sources using SPA compared with OGSBI-SVD and the CRLB, with $N=200$. The horizontal dashed lines refer to lower bounds of on-grid methods with $\widetilde{N}=100,200$.} \label{Fig:MSE_OGSBI}
\end{figure}

In \emph{Experiment 5}, we repeat the simulation by fixing $\text{SNR}=10\text{dB}$ and varying the number of snapshots $N$ from 1 to 200 at a step of 3. Simulation results presented in Fig. \ref{Fig:MSE_OGSBI_varyN} show that SPA has consistently satisfactory performance though it is slightly worse than OGSBI-SVD (the MSEs of SPA are less than $1.5$ times those of OGSBI-SVD in most of the scenarios). Since OGSBI-SVD mimics the ML estimation and SPA is a large-snapshot realization of the ML, the performance gap between them vanishes as $N$ increases.

\begin{figure}
\centering
  \includegraphics[width=3.5in]{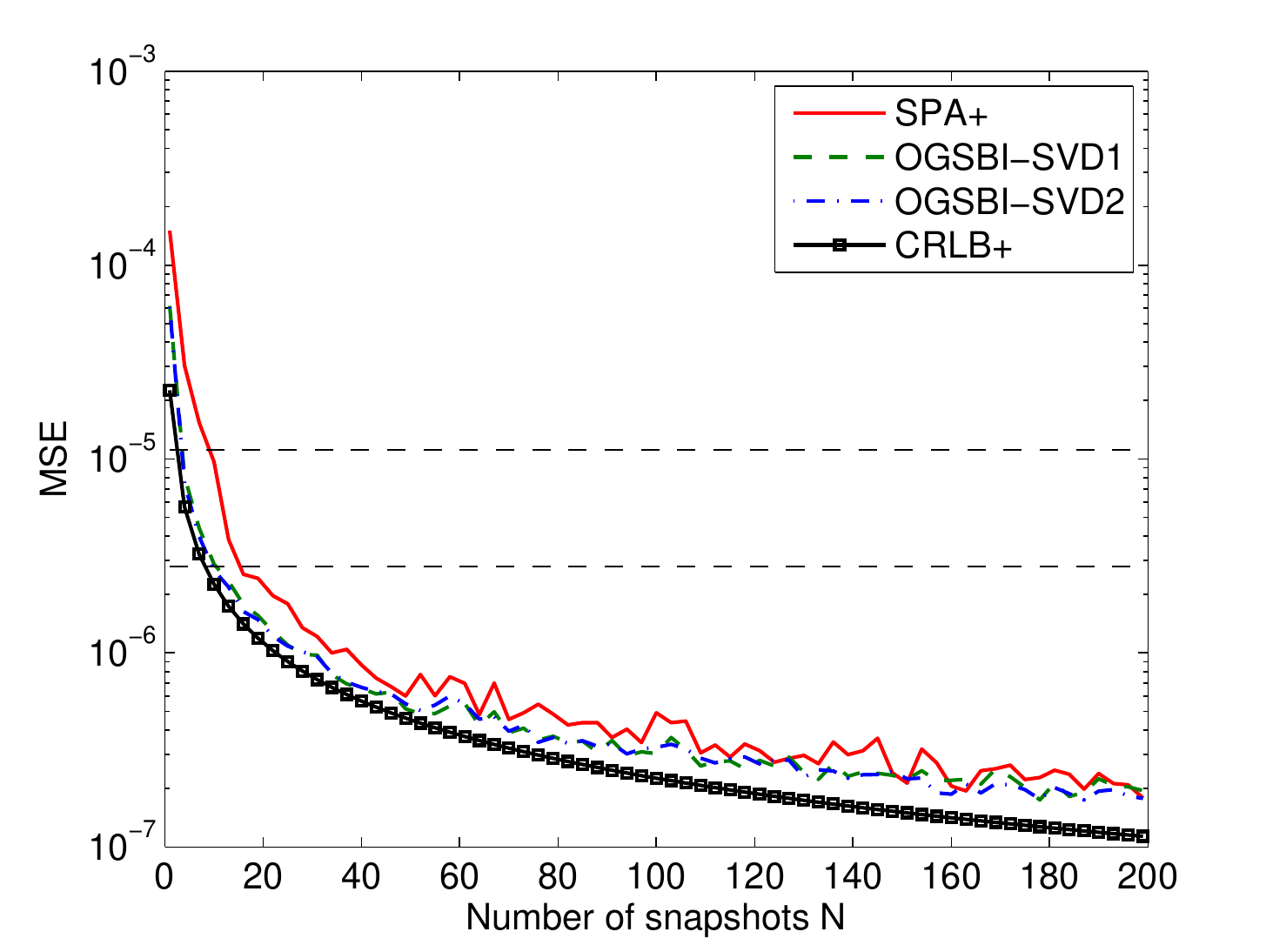}
\centering
\caption{MSEs of frequency estimates of $K=2$ uncorrelated sources using SPA compared with OGSBI-SVD and the CRLB, with $\text{SNR}=10\text{dB}$. The horizontal dashed lines refer to lower bounds of on-grid methods with $\widetilde{N}=100,200$.} \label{Fig:MSE_OGSBI_varyN}
\end{figure}

\subsection{Discussion: Resolution vs. Robustness} \label{sec:discussion}
If an infinitely fine discretization scheme is allowed in SPICE (to eliminate the modeling error), then SPICE and SPA will produce the same covariance matrix estimate $\widehat{\m{R}}$. However, SPA always has a sparse spectrum while SPICE may produce a dense one. The two spectra can be considered as two different decomposition schemes applied to $\widehat{\m{R}}$ to estimate the source powers. The sparse decomposition of SPA has good statistical properties as presented in Subsection \ref{sec:statproperty} while the dense decomposition of SPICE does not. Moreover, as shown in Figs. \ref{Fig:spectra3} and \ref{Fig:spectra_coh2}, the dense decomposition of SPICE has potentially inferior resolution. However, the simulation results of Fig. \ref{Fig:MSE_ULA} show that SPICE can outperform SPA in terms of MSE in the middle range of the SNR (similar results are presented in Figs. \ref{Fig:MSE_OGSBI} and \ref{Fig:MSE_OGSBI_varyN} for OGSBI-SVD). It is because that 1) SPA is empirically observed to produce a more heavy-tailed frequency estimator than SPICE since the frequency of a source is determined by a single point unlike SPICE for which the frequency estimate is given by the peaks of the spectrum,\footnote{The use of peaks of a dense spectrum as the frequency estimates is inspired by spectral-based methods, e.g, MUSIC and beamformer, and has been adopted in many ``sparse'' methods (see, e.g., \cite{stoica2011spice,yang2013off}). However, it conflicts with the principle of sparse methods that any (significant) nonzero power estimate corresponds to a source (ruling out numerical effects). This problem has been neglected in the literature and should be investigated in the future.} and 2) the MSE metric is sensitive to heavy-tailed estimators (in fact, a careful study reveals that SPA+ is more accurate than SPICE3+ in over half of the Monte Carlo runs in the SNR range $\mbra{-10,5}$dB in Fig. \ref{Fig:MSE_ULA}).

\section{Conclusion and Future Work} \label{sec:conclusion}

In this paper, the linear array signal processing problem was studied and a discretization-free technique named as SPA was proposed. The new method adopts the covariance fitting criteria of SPICE and was formulated as an SDP followed by a postprocessing technique. SPA is a parametric method and guarantees to produce a sparse parameter estimate and in the mean time enjoys several other merits. Its asymptotic statistical properties were analyzed and practical performance was demonstrated via simulations compared to existing methods.

The following directions should be investigated in the future, some of which have been mentioned in the main context of this paper: 1) connection of the proposed SPA method to the atomic norm-based discretization-free methods in \cite{candes2013towards,bhaskar2013atomic,tang2012compressed} for spectral analysis, 2) postprocessing techniques for general but not redundancy SLAs to utilize the full information of the range of the source number, 3) performance analysis of SPA in the finite-snapshot case while this paper is mainly on its asymptotic statistical properties, 4) fast implementations of the SPA method via developing more computationally efficient algorithms for solving the SDPs involved, 5) modified SPA methods with automatic source number estimation, and 6) discretization-free methods for array processing with general array geometries.


\end{document}